\documentclass[12pt,oneside,reqno]{amsart}

\usepackage[usenames,dvipsnames,svgnames,table]{xcolor}
\usepackage{amsfonts,amsthm,amsmath,amssymb}
\usepackage{enumerate}
\usepackage{vmargin}
\setpapersize{USletter}
\setmargrb{2.5cm}{2cm}{2.5cm}{2.cm} % --- sets all four margins. Cool.
\newtheorem{thm}{Theorem}
\numberwithin{thm}{section}
\newtheorem{prop}[thm]{Proposition}
\newtheorem{cor}[thm]{Corollary}
\newtheorem{lem}[thm]{Lemma}

\theoremstyle{definition}
 
\newtheorem{remark}[thm]{Remark} 
\newtheorem{defi}[thm]{Definition} 
\newcommand{\tab}{\textsf{Tab}}
\newcommand{\rst}{\textsf{RST}}
\renewcommand{\S}{\mathbb{S}}
\newcommand{\Y}{\mathbb{Y}}
\newcommand{\ct}[2]{\textsf{CT}_{\mathbb{#1}}[#2]}
\renewcommand{\c}[1]{\textsf{CT}_{\mathbb{#1}}}
\newcommand{\row}[2]{\textsf{row}_{\mathbb{#1}}[#2]}
\newcommand{\col}[2]{\textsf{col}_{\mathbb{#1}}[#2]}
\newcommand{\pars}[1]{\textsf{Par}\left(#1\right)}
\newcommand{\comps}{\textsf{Comp}}
\newcommand{\inv}{\mathrm{inv}}
\usepackage{xargs}
\usepackage[colorinlistoftodos,prependcaption,textsize=tiny,linecolor=red,backgroundcolor=red!25,bordercolor=red]{todonotes}
\newcommandx{\charles}[2][1=]{\todo[linecolor=blue,backgroundcolor=blue!25,bordercolor=blue,#1]{#2 ---Charles}}
\newcommandx{\laura}[2][1=]{\todo[linecolor=red,backgroundcolor=red!25,bordercolor=red,#1]{#2 ---Laura}}
\newcommandx{\jeangabriel}[2][1=]{\todo[linecolor=green,backgroundcolor=green!25,bordercolor=green,#1]{#2 ---Jean Gabriel}}
\title[Vector-valued and highest weight Jack and Macdonald polynomials]{Connections between Vector-valued and highest weight Jack and Macdonald polynomials}
\author[Colmenarejo]{Laura Colmenarejo}
\address{Laura Colmenarejo, University of Massachusetts at Amherst, US}
\email{laura.colmenarejo.hernando@gmail.com}
\urladdr{https://sites.google.com/view/l-colmenarejo/home}

\author[Dunkl]{Charles F. Dunkl}
\address{Charles F. Dunkl, Department of Mathematics, University of Virginia,Charlottesville VA 22904-4137, USA}
\email{cfd5z@virginia.edu}
\urladdr{http://people.virginia.edu/$\sim$cfd5z/}

\author[Luque]{Jean-Gabriel Luque}
\address{Jean-Gabriel Luque, Universit\'e de Rouen Normandie,
Laboratoire d’Informatique, du Traitement de l’Information et des Syst\`emes (LITIS), Avenue de l'Universit\'e - BP 8, 76801 Saint-\'Etienne-du-Rouvray Cedex, France}
\email{jean-gabriel.luque@univ-rouen.fr}
\thanks{This work is partially supported by the ``European Regional Development Fund'' (ERDF) via the regional (GRR) project MOUSTIC}

\begin{document}

\begin{abstract}
We analyze conditions under which a projection from the vector-valued Jack or Macdonald polynomials to scalar polynomials has useful properties, specially commuting with the actions of the symmetric group or Hecke algebra, respectively, and with the Cherednik operators for which these polynomials are eigenfunctions. In the framework of representation theory of the symmetric group and the Hecke algebra, we study the relation between singular nonsymmetric Jack and Macdonald polynomials and highest weight symmetric Jack and Macdonald polynomials. Moreover, we study the quasistaircase partition as a continuation of our study on the conjectures of Bernevig and Haldane on clustering properties of symmetric Jack polynomials. 
\end{abstract}

\maketitle

\setcounter{tocdepth}{3}
\tableofcontents

%%%%%%%%%%%%%%%%%%%%%%%%%%%%%%%%%%%%%%%%%%%%%%%%%%%%%%%%%%%%%%%%%%%%%%%%%%%%%%%%%
\section{Introduction}
A \emph{singular polynomial} is a polynomial belonging to the kernel of all the Dunkl operators \cite{D2005}. These operators are of prime importance in the study of polynomial representations of Hecke algebras and, in particular, in the study of Macdonald polynomials (see \cite{Macd1995}). Macdonald polynomials are two parameters $(q,t)$ multivariate polynomials indexed by compositions or partitions which can be symmetric or nonsymmetric and homogeneous or shifted, and which degenerate to (one parameter $\kappa$) Jack polynomials when setting $q=t^{\kappa}$ and sending $t$ to $1$. For certain values of the parameters $(q,t)$ (respectively $\kappa$), some nonsymmetric Macdonald (respectively Jack) polynomials indexed by special compositions are singular. The singular Macdonald polynomials have important properties; for instance, singular shifted Macdonald polynomial equals the corresponding homogeneous one. The notion of singularity has a symmetric counterpart called \emph{highest weight}. Macdonald polynomials are highest weight if they belong to the kernel of the sum of the Dunkl operators. For Jack polynomials, the highest weight polynomials are characterized similarly in terms of the partial sums. 
We propose to investigate the relationship between the notions of singularity and highest weight by using a more general class of polynomials, called \emph{vector-valued} Macdonald and Jack polynomials \cite{DL2011,DL2012}, whose coefficients belong to a representation of the Hecke algebra and which project to the classical case while preserving singularity properties. This work is part of a larger study \cite{CDL2017,DL2015,JL} on the conjectures of Bernevig and Haldane \cite{BH2008} on clustering properties of (symmetric) Jack polynomials.
These clustering properties are closely related to the quasistaircase partition, which get our attention later on the paper. Taking the Jack polynomials case as a guide, we replicate the study for Macdonald polynomials. Our study is based on Theorem \ref{Conjecture}, which is already included in \cite{CDL2017} as a conjecture as part of our conclusions and that will be proved in a forthcoming paper of the authors, \cite{CD2019}.
 
\vspace{0.4cm}

The paper is organized as follows. In Section {\color{blue}\ref{sec:representation theory}}, we set up the representation theory framework for the symmetric group, the Hecke algebra and the polynomials representations. This section also includes a very brief description of the combinatorial objects appearing in our study. In Section {\color{blue}\ref{sec:Vec-Valued}}, we introduce the vector-valued polynomials together with their projections and the characterization of the symmetric elements in the general setting of polynomials. In Section {\color{blue}\ref{sec:JackMacdo}}, Jack and Macdonald polynomials are defined and we state some consequences related to the previous section. Section {\color{blue}\ref{SectionQSP}} is devoted to the quasistaircase partition. Starting with the importance of the quasistaircase and its definition, we analyze the singularity of nonsymmetric Jack polynomials and its consequences for the highest weight symmetric Jack polynomials. In parallel fashion, we finish this section with the analogous analysis for the Macdonald polynomials. In Section {\color{blue}\ref{SectionSP}}, we investigate some consequences of our results and we exhibit some factorization formulas. We finish this paper wrapping up our conclusions and perspective in Section {\color{blue}\ref{sec:conclusion}}. 

\section{Representation theory}\label{sec:representation theory}
Jack and Macdonald polynomials are constructed by use of the representation theory of the symmetric group and the Hecke algebra. In this section we refresh basic concepts and results with the aim of setting up the notation and the starting point.

\subsection{Combinatorial objects}
Let us start recalling the definition of combinatorial objects that appear in our study and setting our notation.

A \emph{partition} $\tau=(\tau_1,\dots,\tau_N)$ is a nonincreasing sequence such that $\tau_i\geq 0$, for all $i$. The length of a partition $\tau$ is the number of nonzero parts of $\tau$, $\ell(\tau)=\max\{i:\ \tau_i>0\}$. Moreover, we say that $\tau$ is a partition of $n$, or that the size of $\tau$ is $n$, if $\sum_i \tau_i=n$. We denote by $\tau \vdash n$ or $|\tau|=n$ if $\tau$ is a partition of $n$ and by $\pars{n}$ the set of partitions of $n$. We consider the following \emph{partial order} on partitions: For $\tau,\gamma \in \pars{n}$, we say that $\tau$ dominates $\gamma$, and we write $\tau \succ \gamma$, if $\tau\neq \gamma$ and $\displaystyle{\sum_{i=1}^j \tau_i \geq \sum_{i=1}^j \gamma_i}$, for all $1\leq j\leq n$.

A \emph{composition} $\alpha=(\alpha_1,\dots,\alpha_N)$ is any permutation of a partition. We denote by $\alpha^+$ the unique nonincreasing rearrangement of $\alpha$ such that $\alpha^+$ is a partition and by $|\alpha|$ the size of the composition, that is, the sum of all its parts. We denote by $\comps$ the set of compositions.

The definition of the partial order on partitions applies also for compositions since it does not use that the sequences are weakly decreasing. Moreover, it can be used to define another order: For $\alpha$ and $\beta$ compositions, we write $\alpha \triangleright \beta$ if $|\alpha|=|\beta|$ and either $\alpha^+\succ \beta^+$, or $\alpha^+=\beta^+$ and $\alpha\succ\beta$.

\begin{remark}
Notice that, by definition, the partitions and compositions appearing in this paper are allowed to have zeros and are standarized to have $N$ entries in total (including the zeros). However, we omit the zero entries in those partitions for which they are not relevant. Moreover, we mostly work with partitions of $N$, $\pars{N}$. That is, the set of partitions $\tau=(\tau_1,\dots,\tau_N)$ with $\sum_i \tau_i = N$. 
\end{remark}

A \emph{Ferrers diagram of shape} $\tau\in\pars{n}$ is obtained by drawing boxes at points $(i,j)$, for $1\leq i \leq \ell(\tau)$ and $1\leq j\leq \tau_i$ (corresponding to French notation). Given a Ferrers diagram of shape $\tau$, we define the \emph{conjugate partition} $\tau^\prime$ as the partition associated to the diagram that is obtained by exchanging rows and columns.

We define three fillings of a Ferrers diagram of shape $\tau\in \pars{n}$:
\begin{itemize}
\item A \emph{column-strict Young tableau} is a filling such that the entries are increasing in the columns and nondecreasing in the rows, and they are in $\{1,2,\dots,n\}$. 

\item A \emph{reverse standard Young tableau} (RSYT) is a filling such that the entries are exactly $\{1,2,\dots,n\}$ and are decreasing in rows and columns. 

\item A \emph{reverse row-ordered standard tableau} is a filling such that the entries are exactly $\{1,2,\dots,n\}$ and are decreasing in rows, with no condition on the columns. 
\end{itemize}

This paper has the reverse standard Young tableaux as one of the main combinatorial objects. Therefore, we denote by $\tab_\tau$ the set of RSYT of shape $\tau$ and let $V_\tau$ be the space spanned by RSYTs of shape $\tau$ with orthogonal basis $\tab_\tau$. Note that $\tab_\tau \subset \rst_\tau$, where $\rst_\tau$ denotes the set of reverse row-ordered standard tableaux of shape $\tau$. 

We finish this subsection introducing useful notation for the tableaux in $\tab_\tau$. 
Let $\S \in \tab_\tau$. The entry $i$ of $\S$ is at coordinates $(\row{S}{i},\col{S}{i})$ and the \emph{content} of the entry is $\ct{S}{i}:= \col{S}{i}-\row{S}{i}$. Then, each $\S\in \tab_\tau$ is uniquely determined by its \emph{content vector} $\c{S} = \left[\ct{S}{i}\right]_{i=1}^N$. 
For instance,
$\S=\scalebox{0.7}{
\begin{tikzpicture}
\draw (0,0) rectangle (0.5,0.5);
\draw (0,0.5) rectangle (0.5,1);
\draw (0.5,0) rectangle (1,0.5);
\draw (0.5,0.5) rectangle (1,1);
\draw (1,0) rectangle (1.5,0.5);
\draw (1,0.5) rectangle (1.5,1);
\draw (1.5,0) rectangle (2,0.5);
\node at (0.25,0.25) {7};
\node at (0.75,0.25) {6};
\node at (1.25,0.25) {5};
\node at (1.75,0.25) {2};
\node at (0.25,0.75) {4};
\node at (0.75,0.75) {3};
\node at (1.25,0.75) {1};
\end{tikzpicture}}$\ has shape $\tau=(4,3)$ and content vector $\c{S}=[1,3,0,-1,2,1,0]$.

There is a partial order on $\tab_\tau$ related to the inversion number:
\begin{eqnarray*}
\text{inv}(\S ) := \# \{ (i,j):\ 1\leq i<j \leq N,\ \ \ct{S}{i}\geq \ct{S}{j}+2 \}.
\end{eqnarray*}
We denote by $\S _0$ the \emph{inv-maximal element} of $\tab_\tau$, which has the numbers $N,N-1,\dots,1$ entered column-by-column, and by $\S _1$ the \emph{inv-minimal element} of $\tab_\tau$, which has these numbers entered row-by-row. 

%A \emph{reverse lattice permutation} $\alpha=\left(\alpha_1,\alpha_2,\dots, \alpha_N\right)$ is a sequence of integers such that in every final part of the sequence any number $\alpha_i$ occurs at least as often as the number $\alpha_{i+1}$. 
Given $\S \in \tab_\tau$, and two nonnegative integers, $m$ and $m_0$, such that $m_0$ is a factor of $m$, we associate to $\S$ the sequence $\alpha(\S)$ defined by setting
\begin{eqnarray}\label{defJsing}
\alpha(\S )_i = v\left(\row{S}{i}\right),\ \text{ for }1\leq i\leq N \text{ and }\S \in\tab_\tau,
\end{eqnarray}
where $v$ is the auxiliary function $v$: $v(1)=0$, and $v(j)=m+(j-2)m_0$, for $2\leq j\leq \ell(\tau)$. The sequence $\alpha(\S)$ has the property that in every final part of the sequence any number $\alpha(\S)_i$ occurs at least as often as the number $\alpha(\S)_{i+1}$. This property resembles to the definition of reverse lattice permutations, also known as Yamanouchi words \cite{Macd1995}, and that is why we call $\alpha(\S)$ the \emph{reverse lattice permutation} of $\S$.

For instance, for $m=2$, $m_0=1$ and 
 $\S =\scalebox{0.7}{
 \begin{tikzpicture}
\draw (0,0) rectangle (0.5,0.5);
\draw (0.5,0) rectangle (1,0.5);
\draw (1,0) rectangle (1.5,0.5);
\draw (1.5,0) rectangle (2,0.5);
\draw (2,0) rectangle (2.5,0.5);
\draw (0,0.5) rectangle (0.5,1);
\draw (0.5,0.5) rectangle (1,1);
\draw (0,1) rectangle (0.5,1.5);
\draw (0.5,1) rectangle (1,1.5);
\draw (0,1.5) rectangle (0.5,2);
%\draw (1.5,0) rectangle (2,0.5);
\node at (0.25,0.25) {10};
\node at (0.75,0.25) {6};
\node at (1.25,0.25) {5};
\node at (1.75,0.25) {3};
\node at (0.25,0.75) {9};
\node at (2.25,0.25) {1};
\node at (0.75,0.75) {4};
\node at (0.25,1.25) {8};
\node at (0.75,1.25) {2};
\node at (0.25,1.75) {7};
\end{tikzpicture}}$, we have 
$\alpha(\S )=(0,3,0,2,0,0,4,3,2,0)$. We do not specify the integers $m$ and $m_0$ in the notation since it will be clear by the context. 

\subsection{The symmetric group and the Hecke algebra}
Let $\mathcal{S}_N$ be the \emph{symmetric group}. That is, it is the set of permutations of $\{1,2,\dots,N\}$, which acts on $\mathbb{C}^N$ by permutation of the coordinates. The group $\mathcal{S}_N$ is generated by simple reflections $s_i:=(i,i+1)$, for $1\leq i<N$, and it is abstractly presented by $\{s_i^2=1:\ 1\leq i<N\}$, together with the \emph{braid relations} ($s_{i+1}s_{i}s_{i+1}=s_{i}s_{i+1}s_{i}$ and $s_{i}s_{j}=s_{j}s_{i}$ for $|i-j|>1$). The associated group algebra $\mathbb{C}\mathcal{S}_N$ is the linear space $\displaystyle{\left\{ \sum_{\sigma \in \mathcal{S}_N} c_\sigma \sigma\right\}}$.
The \emph{Jucys-Murphy elements} for $\mathbb{C}\mathcal{S}_N$ are the elements of the form $\displaystyle{\omega_i :=\sum_{j=i+1}^N s_{ij}}$, for $1\leq i<N$, where $s_{ij}$ denotes the transposition $(i,j)$. 

The symmetric group can be seen as a particular case of the \emph{Hecke algebra}, $\mathcal{H}_N(t)$, defined for a formal parameter $t$ as the associative algebra generated by $\{T_1,T_2,\dots, T_{N-1}\}$ subject to the following relations:
\begin{eqnarray*}
\left\{ \begin{array}{ll}
(T_i+1)(T_i-t)= 0, &\text{ for } 1\leq i < N-1, \\[0.1in]
T_iT_{i+1}T_i = T_{i+1}T_iT_{i+1}, &\text{ for } 1 \leq i< N-2, \text{ and } \\[0.1in]
T_iT_j = T_jT_i,&\text{ for } 1\leq i<j-1\leq N-2.
\end{array}\right.
\end{eqnarray*}
The \emph{Jucys-Murphy elements} for $\mathcal{H}_N(t)$ are defined recursively by $\phi_N = 1$ and $\displaystyle{\phi_i=\frac{1}{t}T_i\phi_{i+1}T_i}$, for $1\leq i <N$.
For $t=1$, $\mathbb{C}\mathcal{S}_N$ and $\mathcal{H}_N(t)$ are identical. Otherwise, there is a linear isomorphism based on the map $s_i \longmapsto T_i$.

\subsection{Irreducible representations}
We start this subsection recalling the definitions of a representation and an irreducible representation in the general framework of group theory. 

A \emph{representation of a group $G$} on a vector space $V$ is a group homomorphism from $G$ to $GL(V)$, the general linear group on $V$. It is common practice to refer to $V$ itself as the representation or $G$-module. A subspace $W$ of $V$ that is invariant under the group action is called a \emph{subrepresentation}. If $V$ has exactly two subrepresentations, namely the zero-dimensional subspace and $V$ itself, then the representation is said to be \emph{irreducible}. 

We describe now the irreducible representations for each case. For the Hecke algebra, the irreducible representations are indexed by partitions of $N$ and, by abuse of notation, we denote them by its indexing partition $\tau \vdash N$. The representations are constructed in terms of RSYT and the action of $\{T_i\}$ on the basis elements. For $\S \in \tab_\tau$ and $i$, with $1\leq i <N$,
\begin{enumerate}[(I)]
\item If $\row{S}{i}=\row{S}{i+1}$, then $\S \tau(T_i)=t\S $.
\item If $\col{S}{i}=\col{S}{i+1}$, then $\S \tau(T_i)=-\S $.
\item If $\row{S}{i}<\row{S}{i+1}$ and $\col{S}{i}>\col{S}{i+1}$, we denote by $\S^{(i)}\in \tab_\tau$ the tableau obtained from $\S $ by exchanging $i$ and $i+1$ and $b=\ct{S}{i}-\ct{S}{i+1}$. Then, $\displaystyle{\S \tau(T_i) = \S ^{(i)}+\frac{t-1}{1-t^{-b}}\S }$.
\item If $\ct{S}{i}-\ct{S}{i+1}\leq -2$, then $\displaystyle{\S \tau(T_i) = \frac{t(t^{b+1}-1)(t^{b-1}-1)}{(t^b-1)^2}\S ^{(i)}+\frac{t^b(t-1)}{t^b-1}\S }$, where we use the same notation than in the previous case. 
\end{enumerate}
Observe that the last case can be obtained from the previous case by interchanging $\S $ and $\S^{(i)}$ and applying the quadratic relation $\tau(T_i)^2=(t-1)\tau(T_i) + t I$, where $I$ denotes the identity operator on $V_\tau$. We will refer to the formulas (I--IV) as the \emph{action formulas for $\tau(T_i)$}. 

Taking $t=1$, we recover the irreducible representation for the symmetric group $\mathcal{S}_N$ and the action of $\tau(s_i)$. We describe them here. For $\S \in \tab_\tau$ and $i$, with $1\leq i <N$, there are four possibilities:
\begin{enumerate}[(I)]
\item If $\row{S}{i}=\row{S}{i+1}$, then $ \S \tau (s_i)=\S $.
\item If $\col{S}{i}=\col{S}{i+1}$, then $\S \tau (s_i)=-\S $.
\item If $\row{S}{i}<\row{S}{i+1}$ and $\col{S}{i}>\col{S}{i+1}$, then $\displaystyle{\S \tau (s_i)= \S ^{(i)}+\frac{1}{b}\S }$, where $b=\ct{S}{i}-\ct{S}{i+1}$ as in the Hecke algebra case. 
\item If $\ct{S}{i}-\ct{S}{i+1}\leq -2$, then $\displaystyle{\S \tau (s_i)= \left(1-\frac{1}{b^2}\right)\S ^{(i)}+\frac{1}{b}\S } $.
\end{enumerate}
We will refer to these formulas as the \emph{action formulas for $\tau(s_i)$}. 

For instance, the Jucys-Murphy elements for $\mathbb{C}\mathcal{S}_N$ act on $\S \in \tab_\tau$ by $\S \tau(\omega_i) =\ct{S}{i}\S $, whereas for $\mathcal{H}_N(t)$, the Jucys-Murphy elements act on $\S \in \tab_\tau$ by $\S \tau(\phi_i)=t^{\ct{S}{i}}\S $.

Consider the following inner product on $V_\tau$: For $\S , \S ^\prime\in\tab_\tau$, $\langle \S ,\S ^{\prime}\rangle_{t}:=\delta_{\S ,\S ^{\prime}}\cdot \gamma(\S ,t)$, with 
\begin{eqnarray*}
\gamma(\S ;t):= \prod_{\substack{i<j \\ \ct{S}{j}-\ct{S}{i}\geq 2}} \frac{\left(1- t^{\ct{S}{j}-\ct{S}{i}-1} \right)\left(1-t^{\ct{S}{j}-\ct{S}{i}+1}\right)}{\left(1-t^{\ct{S}{j}-\ct{S}{i}} \right)^2}.
\end{eqnarray*}
This inner product degenerates to the inner product in $\mathcal{S}_N$: 
\begin{eqnarray*}
\left\langle \S ,\S ^\prime\right\rangle_1 := \lim_{t\rightarrow 1}\langle \S ,\S ^{\prime}\rangle_t= \delta_{\S ,\S ^\prime} \cdot \prod_{\substack{1\leq i<j\leq N \\ \ct{S}{i}\leq \ct{S}{j}-2}} \left(1-\frac{1}{\left[\ct{S}{i}-\ct{S}{j}\right]^2}\right).
\end{eqnarray*}

\subsection{Polynomial representations}\label{SubSectPR}
For $N\geq 2$, let us denote by $x$ the set of variables $\{x_1,x_2,\dots,x_N\}$. For a composition $\alpha \in \comps$, let $x^\alpha=\prod_{i} x_i^{\alpha_i}$ be a monomial of degree $n=|\alpha|$. Let $\mathbb{F}$ be an extension field of $\mathbb{C}$, possibly $\mathbb{C}(\kappa)$ or $\mathbb{C}(q,t)$, for $\kappa$, $q$ and $t$ transcendental or formal parameters. Consider the space of polynomials, $\mathcal{P} :=\text{span}_\mathbb{F} \{ x^\alpha:\alpha\in \comps\}$ and the space of homogeneous polynomials $\mathcal{P}_n :=\text{span}_\mathbb{F} \{ x^\alpha:\ \alpha\in \comps,\ |\alpha|=n\}$, for $n\in\mathbb{Z}_{>0}$.

The action of $\mathcal{S}_N$ on polynomials is defined by saying that $s_i$ permutes the variables $x_i$ and $x_{i+1}$. Therefore, $p(x)s_i:=p(xs_i)$, for $1\leq i<N$. For arbitrary transpositions, $s_{ij}$ exchanges the positions of $x_i$ and $x_j$ and $p(x)s_{ij} := p(xs_{ij})$. In general, $p(x)\sigma=p(x\sigma^{-1})$, where $(x\sigma)_i = x_{\sigma^{-1}(i)}$, since it is an action on the right.

The action of the Hecke algebra $\mathcal{H}_N(t)$ on polynomials is defined by 
\begin{eqnarray*}
p(x)T_i := (1-t)x_{i+1}\frac{p(x)-p(xs_i)}{x_i-x_{i+1}} +tp(xs_i).
\end{eqnarray*}
The defining relations can be verified straightforwardly. 

The following concept will have a relevant role throughout the paper since it characterizes the space of polynomials that transforms according to a specific $\tau$.
\begin{defi}\label{BasisG}
Suppose $\mathcal{V}$ is a linear space of polynomials which is invariant under the actions of $\mathcal{S}_N$ or $\mathcal{H}_N(t)$. We say that the elements of $\mathcal{V}$ are of \emph{isotype} $\tau$ if there is a basis $\{g_\S :\ \S \in\tab_\tau\}$ which transforms under the action formulas for $\tau(s_i)$ or $\tau(T_i)$, respectively. 
\end{defi} 

The following construction shows a canonical space of isotype $\tau$ for the action of $\mathcal{H}_N(t)$ whose elements are of minimal degree. Setting $t=1$ provides the definition for the action of $\mathcal{S}_N$. First, we define the polynomial associated to $\S_0$ as 
\begin{eqnarray}\label{Eqp_S0}
h_{\S_{0}}(x)=\prod_{\substack{1\leq i<j\leq N\\ \textsf{col}_{\mathbb{S}_0}[i] =\textsf{col}_{\mathbb{S}_0}[j] }} \left(tx_{i}-x_{j}\right).
\end{eqnarray}
Now, if $\row{S}{i}< \row{S}{i+1}$ and $\col{S}{i}> \col{S}{i+1}$, then 
\begin{eqnarray}\label{hSgroup}
h_{\S}\left(x\right)T_{i}&=&h_{\S^{(i)}}(x)+\frac{t-1}{1-t^{\ct{S}{i+1}- \ct{S}{i}}}h_{\S}(x), \text{ and } \\
h_{\S}\left(x\right)s_{i}&=&h_{\S^{(i)}}(x)+\frac{1}{\ct{S}{i}- \ct{S}{i+}}h_{\S}(x), \nonumber
\end{eqnarray}
respectively. Since $\inv\left(\S^{(i)}\right)=\inv(\S) +1$, $h_{\S}$ is determined by $\left\{ h_{\S^{\prime}}:\inv\left(S^{^{\prime}}\right)>\inv\left(S\right)\right\}$ and can be computed from $h_{\S_{0}}$ applying successively \eqref{hSgroup}. %The polynomials associated to a RSYT of shape $\tau$ have degree $\sum_i (i-1)\tau_i$. 

For example, consider $\tau=(2,1)$, for which there are two RSYT of this shape, 
$\S _0=\scalebox{0.7}{
\begin{tikzpicture}
\draw (0,0) rectangle (0.5,0.5);
\draw (0,0.5) rectangle (0.5,1);
\draw (0.5,0) rectangle (1,0.5);
\node at (0.25,0.25) {3};
\node at (0.75,0.25) {1};
\node at (0.25,0.75) {2};
\end{tikzpicture}}$ and $\S _1=\scalebox{0.7}{
\begin{tikzpicture}
\draw (0,0) rectangle (0.5,0.5);
\draw (0,0.5) rectangle (0.5,1);
\draw (0.5,0) rectangle (1,0.5);
\node at (0.25,0.25) {3};
\node at (0.75,0.25) {2};
\node at (0.25,0.75) {1};
\end{tikzpicture}}$ .
By \eqref{Eqp_S0}, $h_{\S _0}(x) = tx_2-x_3$. Moreover, by \eqref{hSgroup},
$h_{\S_0}(x)T_1 = (t-1)tx_2 + t^2x_1 - tx_3$, and so $\displaystyle{h_{\S _1}(x)=t^2x_1 - \frac{t}{t+1}(x_2+x_3)}$.

\begin{remark}
	The dominant monomial in $h_{\S_{0}}$ is the product of the dominant monomials of the Vandermonde determinants associated to the columns. For instance consider
	${}
	\S_0 =\scalebox{0.7}{\begin{tikzpicture}
\draw (0,0) rectangle (0.5,0.5);
\draw (0.5,0) rectangle (1,0.5);
\draw (1,0) rectangle (1.5,0.5);
\draw (0,1) rectangle (0.5,1.5);
\draw (0,0.5) rectangle (0.5,1);
\draw (0.5,0.5) rectangle (1,1);
\node at (0.25,0.25) {6};
\node at (0.75,0.25) {3};
\node at (1.25,0.25) {1};
\node at (0.25,1.25) {4};
\node at (0.25,0.75) {5};
\node at (0.75,0.75) {2};
\end{tikzpicture}}
	$, the dominant monomial in $h_{\S_{0}}$ is $x_{4}^{2}x_{5}\cdot x_{2}\cdot 1$.
	
\end{remark}
\subsection{More operators}\label{OpsSection}
In this subsection we present several fundamental families of operators acting on polynomials. The Jack and Macdonald polynomials are defined through them, but these operators are also an important tool for our study.  

Before introducing the operators, it is important to point out that we work with operators acting on the right. It is less common, but it follows the reading direction. For instance, if we look at the partial derivatives, $\displaystyle{\partial_i= \frac{\partial}{\partial x_i}}$, and the divided differences, $\displaystyle{\partial_{ij}= (1-s_{ij})\frac{1}{x_i-x_j}}$ as right operators, then 
\begin{eqnarray*}
x_i\partial_i = 1 + \partial_i x_i \hspace{0.7cm} \text{ and } \hspace{0.7cm} \partial_{ij} = \partial_{ij} x_i + s_{ij}.
\end{eqnarray*}

We start defining the operators for the symmetric group case. Let $\kappa$ be a parameter. 
The \emph{Dunkl operators} are defined for $1\leq i \leq N$ as:
\begin{eqnarray*}
	\mathcal D_{i}=\partial_i+\kappa\sum_{i\neq j}\partial_{ij}, 
\end{eqnarray*}
	and the \emph{Cherednik operators} are defined as:
\begin{eqnarray*}
\mathcal U_{i}=\mathcal D_{i}x_{i}+1+\kappa\omega_{i}.
\end{eqnarray*}
The operators $\mathcal{D}_i$ and $\mathcal{U}_i$ commute with each other. Moreover, the last ones satisfy the following commutation relations: $s_i\mathcal{U}_is_i= \mathcal{U}_{i+1}+\kappa s_i$, for all $1\leq i<N$, and $s_i\mathcal{U}_j = \mathcal{U}_j s_i$, for $i<j$ or $j+1<i$ (see \cite{Macd1995}).

We also have the operator $\omega$ that acts as the product of all the simple transpositions $\omega=s_1s_2\cdots s_{N-1}$.

For the Hecke algebra case, let $q$ and $t$ be parameters. Then, the \emph{$(q,t)$-Dunkl operators} are defined recursively as
\begin{eqnarray*}
\mathcal D^{q,t}_{N} &=& (1-\xi_{N})x_{N}^{-1}, \\[0.1in]
\mathcal D^{q,t}_{i} &=& \frac1tT_{i}\mathcal D_{i+1}^{q,t}T_{i} \ \ \text{ for } 1\leq i< N,
\end{eqnarray*}
and the \emph{$(q,t)$-Cherednik operators} are defined as
\begin{eqnarray*}
\xi_{i} &=& t^{i-1}T_{i-1}^{-1}T_{i-2}^{-1}\cdots T_{1}^{-1}\omega^{q} T_{N-1}T_{N-2}\cdots T_{i}, \hspace{0.3cm} \text{ where } \\[0.1in]
f\omega^{q} &=& f(qx_{N},x_{1},\dots,x_{N-1}).
\end{eqnarray*}
Since the only repeated notation appears in the Dunkl operators, we include the superscript $\ ^{q,t}$ to avoid confusion with the group case.

Note that the $(q,t)$-Cherednik operators satisfy the recursion: $\displaystyle{\xi_{i}=\frac{1}{t}T_{i}\xi_{i+1}T_{i}}$, and that they commute with each other (see \cite{BF}). 

\subsection{Singularity and the properties SP and SMP}
We introduce two definitions for each of our frameworks, and one of the goals of this paper is to relate them. 
The first concept is a property that characterize some polynomials.
\begin{defi}
A polynomial $p$ with rational coefficients is called \emph{singular} in the framework of the symmetric group if there is a rational number $\kappa_0$ such that, for all $i$, $p\mathcal{D}_i=0$ for $\kappa=\kappa_0$. 
Analogously, a polynomial $p$ is \emph{singular} in the framework of the Hecke algebra if, for all $i$, $p\mathcal{D}_{i}=0$ for a specialization of the value of $(q,t)$ (such as $q=t^{-4}$).
\end{defi}
%\begin{remark}
%Note that in the symmetric group case, being singular can be characterized in terms of the sum of the partial derivatives. In the Hecke algebra case, being singular is equivalent to $f\xi_i = f \omega_i$, for all $i$. 
%\end{remark}
Let us illustrate this with an example. Consider $p(x)=\prod \limits_{1\leq i<j\leq3}\left(tx_{i}-x_{j}\right)$ and the specialization $q=t^{-2}$. Then, $p(x)$ is singular since a direct computation shows that $pT_{i}=-p$, for $i=1,2$. In fact, $pT_{1}T_{2}=p$ and $p(x)\omega^q=\left(tqx_{3}-x_{1}\right)\left(tqx_{3}-x_{2}\right)\left(tx_{1}-x_{2}\right)=t^{-2}p(x)$. Furtheremore, by specializing $qt=1$, we obtain directly that $p\xi_3 = p$.

The second concept is related to the bases.
\begin{defi}
We say that a basis $\{g_\S :\ \S\in \tab_\tau\}$ satisfies the \emph{property (SP)} if it is a basis for a space of polynomials of isotype $\tau$ for $\mathcal{S}_N$ and each $g_\S $ is singular, all with $\kappa = \kappa_0$. Analogously, we say that a basis $\left\{g_{\S }:\S \in\tab_\tau\right\}$ satisfies the \emph{property (SMP)} if it is a basis for a space of polynomials of isotype $\tau$ for $\mathcal{H}_{N}\left(t\right)$ and each $g_{\S }$ is singular, all with the same specialization of $(q,t)$. Our study in Sections \ref{SectionQSP} and \ref{SectionSP} are focused on specializations of the form $q^mt^n=1$, with $2\leq n \leq N$, and possibly subject to some other conditions, which are the most common specializations. 
\end{defi}

\section{Vector-valued polynomials\label{sec:Vec-Valued}}

We consider the space of vector-valued polynomials $\mathcal{P}_\tau =\mathbb{C}[x_1,x_2,\dots,x_N] \otimes \mathbb{C}\left[\tab_\tau\right]$, for a given partition $\tau \vdash N$. Then, we describe $\mathcal{P}_\tau = \mathcal{P}\otimes V_\tau$ as the span of $x^\alpha \otimes \S $, for $\alpha\in \comps$ and $\S \in \tab_\tau$.

\subsection{Definitions and action of the operators}
Each of the operators defined in Section \ref{sec:representation theory} has its counterpart in the vector-valued polynomial space. The following table summarizes all the operators already defined, both in the symmetric group and Hecke algebra case, together with their vector-valued version. 
\begin{eqnarray*}
\begin{array}{c|c}
\text{Case of polynomials} & \text{Case of vector-valued polynomials} \\[0.1in] \hline \rule{0pt}{4ex} 
	\partial_i = {\partial\over\partial x_{i}}&\displaystyle{\mathbf {d\over d x_{i}}:={\partial_i}\otimes 1}\\[0.12in]
	\partial_{ij}&\mathbf d_{ij}:=\partial_{ij}\otimes\tau(s_{ij})\\[0.1in]
	\sigma&\mathbf t_{\sigma}:=\sigma\otimes \tau(\sigma) \hspace{0.5cm} \text{ for any } \sigma\in \mathcal{S}_N\\[0.1in]
\omega&\mathbf w:=\omega\otimes \tau(s_{1}s_{2}\dots s_{N-1}) \\[0.1in]\mathcal{D}_i & \displaystyle{\mathbf D_{i}:=\mathbf {d\over d x_{i}}+\kappa\sum_{i\neq j}\mathbf d_{ij}} \\[0.13in]
\omega_i & \displaystyle{\mathbf w_{i}:=\sum_{j>i}\mathbf t_{s_{ij}} }\\[0.12in]
\mathcal{U}_i & \mathbf U_{i}:=\mathbf D_{i}\mathbf x_{i}+1\otimes 1+\kappa\mathbf w_{i}\\[0.12in]
T_i & \mathbf T_{i}:=(1-t)\partial_{i}x_{i+1}\otimes 1+s_{i}\otimes \tau(T_{i}) \\[0.12in]
\omega^q & \mathbf w^{q}:=t^{1-N}\omega^{q}\otimes\tau(T_{1}T_{2}\dots T_{N-1}) \\[0.12in]
\xi_i & \mathbf E_{i}:=t^{i-1}\mathbf T_{i-1}^{-1}\mathbf T_{i-2}^{-1}\cdots \mathbf T_{1}^{-1}\mathbf w^{q} \mathbf T_{N-1}\mathbf T_{N-2}\cdots \mathbf T_{i} \\[0.2in]
\mathcal{D}_N^{q,t} & \mathbf D^{q,t}_{N}:=(1-\mathbf E_{N})x_{N}^{-1}\otimes 1 \\[0.1in]
\mathcal{D}_i^{q,t} & \displaystyle{\mathbf D^{q,t}_{i}:=\frac1t\mathbf T_{i}\mathbf D^{q,t}_{i+1}\mathbf T_{i}}
\end{array}
\end{eqnarray*}

Note that we use \textbf{bold font} for the operators that act on vector-valued polynomials. This way, it will be clearer on which space we act and our results will be more nicely presented.

\subsection{Projections}\label{subsect:proj}
We now set up a projection map from vector-valued polynomials to scalar ones for each framework that also behave well with some of the operators. 
\begin{prop}\label{RhoJack}
Let $\{g_\S :\ \S \in\tab_\tau\}$ be a basis for a space of polynomials of isotype $\tau$ for $\mathcal{S}_N$. The linear map
\begin{eqnarray*}
\begin{array}{rccc}
\rho: & \mathcal{P}_\tau & \longrightarrow & \mathcal{P} \\
& \displaystyle{\sum_{\S \in\tab_\tau}f_\S (x) \otimes \S}& \longmapsto & \displaystyle{\sum_{\S \in\tab_\tau} f_\S (x)g_\S (x)}
\end{array}
\end{eqnarray*}
intertwines $\mathbf{t}_\sigma$ with $\sigma$, for all $\sigma\in\mathcal{S}_N$; that is, $\mathbf{t}_\sigma \rho = \rho \sigma$.
\end{prop}

\begin{proof}
It follows immediately from the definition of $\{g_\S :\ \S \in\tab_\tau\}$. For instance, taking $\sigma=s_i$ and applying the action formulas for $\tau(s_i)$, we have that 
\begin{multline*}
\left(f_\S\otimes \S\right) \mathbf{t}_{s_i} \rho = \left(f_{\S }\left(xs_{i}\right)\otimes\left( \S \tau\left(s_{i}\right)\right)\right)\rho= f_\S (xs_i)\left( g_\S (x) + \frac{1}{b}g_{\S ^{(i)}}(x)\right) =f_{\S }\left(xs_{i}\right)g_{\S }\left(xs_{i}\right) = \\
= \left(f_\S(x)g_\S(x)\right) s_i = \left(f_\S (x)\otimes \S\right) \rho s_i.
\end{multline*}
\end{proof}
 
As before, the basis $\left\{g_{\S }:\S \in\tab_\tau\right\}$ provide us a linear map that, roughly speaking, commutes with $T_i$ and we have the Hecke algebra version of the previous result.
\begin{prop}\label{Hcomm}
Let $\left\{g_{\S }:\S \in\tab_\tau\right\}$ be a basis for a space of polynomials of isotype $\tau$ for
$\mathcal{H}_{N}(t)$. Then, the linear map 
\begin{eqnarray*}
\begin{array}{rccc}
\rho: & \mathcal{P}_\tau & \longrightarrow & \mathcal{P} \\
& \displaystyle{\sum_{\S \in\tab_\tau}f_\S (x) \otimes \S}& \longmapsto & \displaystyle{\sum_{\S \in\tab_\tau} f_\S (x)g_\S (x)}
\end{array}
\end{eqnarray*}
intertwines $\mathbf T_i$ with $T_i$, for $1\leq i<N$; that is, $\mathbf{T}_i \rho = \rho T_i$. 
\end{prop}

\begin{proof}
This result follows from the study of $\left[(f_\S(xs_i)\otimes \S\tau(T_i)\right] \rho$ according to the action of $\tau(T_i)$ and the fact that
\begin{multline*}
\left( f_\S (x)\otimes \S\right) \mathbf{T}_i \rho = \left[(1-t)x_{i+1}\left(f_\S\partial_i\right)\otimes \S + f_\S (xs_i) \otimes \S\tau(T_i) \right] \rho = \\ 
= (1-t)x_{i+1}\left(f_\S\partial_i\right)g_\S(x) + f_\S (xs_i)\left(g_\S T_i\right) = \left[(1-t)x_{i+1}\frac{f_{\S }(x)-f_{\S }\left(xs_{i}\right)}{x_{i}-x_{i+1}}\right]g_{\S }(x) +\\
+f_{\S }\left(xs_{i}\right)\left[(1-t)x_{i+1}\frac{g_{\S }(x)-g_{\S }\left(xs_{i}\right)}{x_{i}-x_{i+1}}+tg\left(xs_{i}\right)\right] = \left(f_{\S }g_{\S }\right)T_{i} = \left(f_\S(x)\otimes \S\right) \rho T_i.
\end{multline*}
\end{proof}

\subsection{Characterization of the symmetric elements}
In this section, we characterize the vector-valued polynomials of $\mathcal{P}_\tau$ that are symmetric under the two actions, for a fixed partition $\tau$ of $N$.

We start with the action of the symmetric group. 
\begin{thm}\label{SGsymm}
Consider $p=\displaystyle{\sum_{\S \in\tab_\tau} p_\S\otimes \S }$, with $p_\S \in \mathcal{P}$ for all $\S $. Then, $ps_i=p$, for all $1\leq i<N$, if and only if, for each $\S \in\tab_\tau$, $p_\S= \displaystyle{\frac{1}{\left\langle \S ,\S \right\rangle_{1}}g_\S }$, where $\{g_\S:\ \S \in\tab_\tau \}$ is a basis for the polynomials of isotype $\tau$.
\end{thm}

\begin{proof}
Assume that $ps_i=p$, for all $1\leq i < N$. This equation can be written, for a fixed $i$, as 
\begin{eqnarray}\label{EqChar1}
\sum_{\S \in \tab_\tau} p_\S (xs_i) \otimes \left(\S \tau(s_i)\right)= \sum_{\S \in \tab_\tau} p_\S (x)\otimes \S .
\end{eqnarray}
As usual in these arguments, we work with invariance for each simple reflection $s_i$, which suffices to prove the symmetry. For that, we analyze the action of $s_i$ according to the action formulas for $\tau(s_i)$ over $p_\S $.
\begin{enumerate}[(I)]
\item For $\row{S}{i}=\row{S}{i+1}$, $p_\S (xs_i) \otimes \left(\S \tau(s_i)\right)=p_\S (xs_i) \otimes \S =p_\S (x)\otimes \S $, thus $p_\S (xs_i)=p_\S(x)$.

\item For $\col{S}{i}=\col{S}{i+1}$, $p_\S (xs_i) \otimes \left(\S \tau(s_i)\right)=-p_\S (xs_i) \otimes \S =-p_\S (x)\otimes \S $, thus $p_\S (xs_i)=-p_\S(x)$.

\item For $\row{S}{i}<\row{S}{i+1}$, $\S ^{(i)}\in\tab_\tau$ and the equation \eqref{EqChar1} requires that 
$$ p_\S (xs_i)\otimes \left(\S \tau(s_i)\right)+ p_{\S ^{(i)}}(xs_i)\otimes \left(\S ^{(i)}\tau(s_i)\right) = p_\S (x)\otimes \S+ p_{\S ^{(i)}}(x)\otimes \S ^{(i)}.$$
Using the action, we have that
\begin{multline*}
p_\S (xs_i)\otimes \left(\S \tau(s_i)\right)+ p_{\S ^{(i)}}(xs_i)\otimes \left(\S ^{(i)}\tau(s_i)\right) = \\ 
= p_\S (xs_i) \otimes \left\{ \S ^{(i)} + \frac{1}{b}\S \right\} + p_{\S ^{(i)}}(xs_i)\otimes \left\{ \left(1-\frac{1}{b^2}\right)\S- \frac{1}{b}\S ^{(i)}\right\} = \\
= \left\{ \frac{1}{b}p_\S (xs_i) + \left(1-\frac{1}{b^2}\right)p_{\S ^{(i)}}(xs_i)\right\} \otimes \S+ \left\{ p_\S (xs_i) -\frac{1}{b}p_{\S ^{(i)}}(xs_i)\right\} \otimes \S ^{(i)}.
\end{multline*}
Matching up the coefficients and rewriting the equations, we obtain that
\begin{eqnarray*}
p_\S (xs_i) &=& \frac{1}{b}p_{\S ^{(i)}}(xs_i) + p_{\S ^{(i)}}(x),\\
p_{\S ^{(i)}} (xs_i) &=& \frac{b^2}{b^2-1} \left( p_\S (x) - \frac{1}{b}p_\S (xs_i)\right),
\end{eqnarray*}
which imply that
\begin{eqnarray*}
p_\S (xs_i) &=& \frac{1}{b}p_\S (x) + \left(1-\frac{1}{b^2}\right) p_{\S ^{(i)}}(x), \\
p_{\S ^{(i)}} (xs_i) &=& p_\S (x) - \frac{1}{b} p_{\S ^{(i)}}(x).
\end{eqnarray*}
Substitute $p_\S =\displaystyle{\frac{1}{\left\langle \S ,\S \right\rangle_1}g_\S }$ and $p_{\S ^{(i)}}=\displaystyle{\frac{1}{\left\langle {\S ^{(i)}},{\S ^{(i)}}\right\rangle_1}g_{\S ^{(i)}}}$, then
\begin{eqnarray*}
g_\S (xs_i)&=& \frac{1}{b}g_\S (x)+\left(1-\frac{1}{b^2}\right) \frac{\left\langle \S ,\S \right\rangle_1}{\left\langle \S ^{(i)},\S ^{(i)}\right\rangle_1} g_{\S ^{(i)}}(x), \\
g_{\S ^{(i)}}(xs_i) &=& \frac{\left\langle \S ^{(i)},\S ^{(i)}\right\rangle_1}{\left\langle \S ,\S \right\rangle_1}g_\S (x) - \frac{1}{b} g_{\S ^{(i)}}(x).
\end{eqnarray*}
By definition, $\displaystyle{\frac{\left\langle \S ^{(i)},\S ^{(i)}\right\rangle_1}{\left\langle \S ,\S \right\rangle_1}=1-\frac{1}{b^2}}$, which shows that $g_\S $ and $g_{\S ^{(i)}}$ transform according to the relations.

\item This last case follows from the previous one by interchanging $\S $ and $\S ^{(i)}$.
\end{enumerate}
This finishes one implication. 
To prove the converse, let $\{g_\S:\ \S \in\tab_\tau\}$ be a basis for the polynomials of isotype $\tau$ such that $p=\displaystyle{\sum_{\S \in\tab_\tau} p_\S\otimes \S }$, where $p_\S= \displaystyle{\frac{1}{\left\langle \S ,\S \right\rangle_{1}}g_\S }$, for each $\S\in\tab_\tau$. Then, the previous study by cases proves that $ps_i=p$, for all $1\leq i<N$, if it is read \emph{backwards}.
\end{proof}

Now, it is the turn for the Hecke algebra. 
\begin{thm}\label{Hsymm}
Consider $p=\displaystyle{\sum_{\S \in\tab_\tau} p_\S\otimes \S }$, with $p_\S \in \mathcal{P}$ for all $\S$. Then, $pT_i=tp$, 
for $1\leq i<N$, if and only if, for each $\S \in\tab_\tau$, $\displaystyle{p_\S = \frac{1}{\langle \S ,\S \rangle_{t}} g_\S }$, where $\{g_\S: \S\in \tab_\tau\}$ is a basis for polynomials of isotype $\tau$ for $\mathcal{H}_N(t^{-1})$. 
\end{thm}

\begin{proof}
Assume that $pT_i=tp$, for a fixed $i$. 
We analyze the action of $\mathbf T_{i}$ according to the action formulas for $\tau(T_i)$ over $p_\S $.
\begin{enumerate}[(I)]
\item For $\row{S}{i}=\row{S}{i+1}$, we have
that $p_\S (xs_i)=p_\S (x)$ since
\begin{multline*}
\left(p_\S (x)\otimes \S \right)\mathbf T_{i}= (1-t)x_{i+1}\frac{p_\S (x) -p_\S (xs_i)}{x_i-x_{i+1}} \otimes \S +p_\S (xs_i)\otimes \S\tau(T_i)= \\
= (1-t)x_{i+1}\frac{p_\S (x) -p_\S (xs_i)}{x_i-x_{i+1}} \otimes \S +tp_\S (xs_i)\otimes \S =tp_\S (x)\otimes \S .
\end{multline*}
\item For $\col{S}{i}=\col{S}{i+1}$, we have that 
\begin{multline*}
(p_\S (x)\otimes \S )\mathbf T_{i}=(1-t)x_{i+1}\frac{p_\S (x) -p_\S (xs_i)}{x_i-x_{i+1}}\otimes \S +p_\S (xs_i)\otimes \S \tau(T_i) =\\
=(1-t)x_{i+1}\frac{p_\S (x) -p_\S (xs_i)}{x_i-x_{i+1}}\otimes \S -p_\S (xs_i)\otimes \S= tp_\S (x)\otimes \S ,
\end{multline*}
which is equivalent to 
$$\left(\frac{1}{t}-1\right)x_{i+1}\frac{p_\S (x)-p_\S (xs_i)}{x_i-x_{i+1}} -\frac{1}{t}p_\S (xs_i) = p_\S (x).$$
That is $p_\S T_i\left(t^{-1}\right)=-p_\S $.

\item For $\row{S}{i}<\row{S}{i+1}$, the idea is to introduce $p\partial_i=\displaystyle{\frac{p(x)-p(xs_i)}{x_i-x_{i+1}}}$ and solve $ps_i = p-(x_i-x_{i+1})(p\partial_i)$. For doing this, we abbreviate $f=p_{\S }$ and $f^{\prime}=p_{\S ^{(i)}}$ and we require that $t\left(f\otimes \S +f^{\prime}\otimes \S ^{(i)}\right) = \left(f\otimes \S +f^{\prime}\otimes \S ^{(i)}\right) \mathbf T_i$. For the right side, we have that
\begin{multline*}
\left(f\otimes \S +f^{\prime}\otimes \S ^{(i)}\right) \mathbf T_i
=\left[(1-t)x_{i+1}\left(f\partial_{i}\right)\right]\otimes \S +\left[(1-t)x_{i+1}\left(f^{\prime}\partial_{i}\right) \right]\otimes \S ^{(i)}\\
 +fs_{i}\otimes\left(\S ^{(i)}+a_{11}\S \right)+f^{\prime }s_{i}\otimes\left(a_{21}\S +a_{22}\S ^{\left(i\right)}\right),
\end{multline*}
where the coefficients $\left\{a_{11},a_{21},a_{22}\right\}$ are given by the action formulas for $\tau(T_i)$. Replace $fs_{i}$ and $f^{\prime}s_{i}$ by the above expressions in $f\partial_{i}$ and $f^{\prime}\partial_{i}$. Matching up the coefficients of $\S $ and $\S ^{\left(i\right)}$ gives two equations which are solved for $f\partial_{i}$ and $f^{\prime}\partial_{i}$:%
\begin{align*}
f\partial_{i} &=\frac{t^{b+1}-1}{\left(t^{b}-1\right)\left(x_{i+1}t-x_{i}\right)}f+\frac{\left(t^{b+1}-1\right)\left(t^{b}-t\right)}{\left(t^{b}-1\right)^{2}\left(x_{i+1}t-x_{i}\right)}f^{\prime},\\
f^{\prime}\partial_{i}&=\frac{1}{x_{i+1}t-x_{i}}f-\frac{t^{b}-1}{\left(x_{i+1}t-x_{i}\right)\left(t^{b}-1\right)}f^{\prime}.
\end{align*}
Substitute these expressions in $fT_{i}\left(t^{-1}\right)$ and $f^{\prime}T_{i}\left(t^{-1}\right)$ to obtain
\begin{align*}
fT_{i}\left(t^{-1}\right)&=-\frac{\left(t-1\right)t^{b}}{\left(t^{b}-1\right)t}f+\frac{\left(t^{b+1}-1\right)\left(t^{b}-t\right)}{t\left(t^{b}-1\right)^{2}}f^{\prime},\\
f^{\prime}T_{i}\left(t^{-1}\right)&=\frac{1}{t}f+\frac{t-1}{\left(t^{b}-1\right)t} f^{\prime}.
\end{align*}
Now, suppose that $f=\c{S}g_{\S }$ and $f^\prime=\textsf{CT}_{\mathbb{S}^{(i)}} g_{\S ^{(i)}}$, then $\left\{ g_{\S }\right\}$ satisfies the hypotheses for isotype $\tau$ in $\mathcal{H}_{N}\left(t^{-1}\right)$ provided $\displaystyle{\frac{\textsf{CT}_{\mathbb{S}^{(i)}}}{\c{S}}=\frac{\gamma\left(\S ;t\right)}{\gamma\left(\S ^{(i)};t\right)}}$.
\item This case follows from the previous one. 
\end{enumerate}
Note that this ratio is invariant under $t\rightarrow t^{-1}$, and one implication is proved. 

For converse, let $\{g_\S: \S\in \tab_\tau\}$ be a basis for the polynomials of isotype $\tau$ such that $p=\displaystyle{\sum_{\S \in\tab_\tau} p_\S\otimes \S }$, with $\displaystyle{p_\S = \frac{1}{\langle \S ,\S \rangle_{t}} g_\S }$, for each $\S\in\tab_\tau$. Then, the previous study by cases proves that $pT_i=tp$, for all $1\leq i<N$, if it is read \emph{backwards}. 
\end{proof}

\subsection{Singularity and projections}
We now explain the relation between the projections and the singular polynomials. For that, we first use the Jucys-Murphy elements to characterize the singular polynomials and then, we describe the relation in each framework. We start with the symmetric group case.
\begin{lem}\label{LemmaJM}
A polynomial $g$ is singular for $\kappa=\kappa_{0}$ if and only if
$g\mathcal{U}_{i}=g+\kappa_{0}g\omega_{i}$, where $\omega_i$ are the Jucys-Murphy elements for $\mathcal{S}_N$.
\end{lem}

\begin{proof}
In general, it holds that $\displaystyle{(x_ig(x))\mathcal{D}_i = g(x) +x_i(g\mathcal{D}_i)+\kappa \sum_{j\neq i} g(xs_{ij})}$. Specializing to $\kappa_0$, we have that $\displaystyle{g\mathcal{U}_i = g(x)+\kappa_0 \sum_{j>i}g(xs_{ij})}$. The converse is proved similarly.
\end{proof}

The following result sets up the key role of the singular polynomials in the projection map. 
\begin{prop}\label{SPcomm}
Suppose $\{g_\S:\ \S \in\tab_\tau\}$ satisfies the \emph{property} \emph{(SP)}. Then, the map $\rho$ defined in Proposition \ref{RhoJack} intertwines $\mathbf{D}_{i}$, $\mathbf{U}_{i}$ with
$\mathcal{D}_{i}$, $\mathcal{U}_{i}$ respectively, for $1\leq i\leq N$. That is, $\mathbf{D}_{i}\rho=\rho\mathcal{D}_{i}$ and $\mathbf{U}_{i}\rho=\rho\mathcal{U}_{i}$, for $1\leq i\leq N$. 
\end{prop}

\begin{proof}
For any $f,g\in\mathcal{P}$ and any $\kappa$, we have the general formula
$$\displaystyle{\left(fg\right)\mathcal{D}_{i}=f\left(g\mathcal{D}_{i}\right)+g\frac{\partial}{\partial x_{i}}f+\kappa\sum_{\substack{j=1\\ j\neq i}}^{N}\frac{f(x)-f(xs_{ij})}{x_{i}-x_{j}}g(xs_{ij})}.$$
Let us specialize $\kappa=\kappa_{0}$ and consider 
$$\left(f_{\S }\otimes \S \right)\mathbf{D}_{i}\rho=\left( f_\S \partial_i \otimes \S + \kappa_0 \sum_{i\neq j} f\partial_{ij} \otimes \S\tau(s_{ij})\right) \rho =\left(f_{\S }\partial_i\right)g_{\S }+\kappa_{0}\sum_{\substack{j=1 \\ j\neq i}}^{N}\frac{f(x)-f(xs_{ij})}{x_{i}-x_{j}}g_{\S ,j}(x),$$ 
where $g_{\S ,j}$ corresponds to $\S \tau(s_{ij})$.
Then, for each $\sigma\in\mathcal{S}_{N}$, there is a matrix $\left[\tau\left(\sigma\right)_{\S ,\S ^{\prime}}\right]$ such that $\S \tau(\sigma)=\sum\limits_{\S ^{\prime}\in\tab_\tau}\S ^{\prime}\tau(\sigma)_{\S ,\S ^{\prime}}$. By definition $g_{\S }(x\sigma)=\sum\limits_{\S ^{\prime}\in\tab_\tau}g_{\S ^{\prime}}(x)\tau(\sigma)_{\S ,\S ^{\prime}}$ and, in particular, $g_{\S ,j}(x)=\sum\limits_{\S ^{\prime}\in\tab_\tau}g_{\S ^{\prime}}(x)\tau (s_{ij})_{\S ,\S ^{\prime}}=g_{\S }(xs_{ij})$. Thus,
\begin{eqnarray*}
\left(f_{\S }\otimes \S \right)\rho\mathcal{D}_{i}=(f_\S g_\S )\mathcal{D}_i=f_{\S }\left(g_{\S }\mathcal{D}_{i}\right)+g_{\S }\frac{\partial}{\partial x_{i}}f_{\S }+\kappa_{0}\sum_{\substack{j=1 \\j\neq i}}^{N}\frac{f_{\S }(x)-f_{\S }(xs_{ij})}{x_{i}-x_{j}}g_{\S }(xs_{ij}).
\end{eqnarray*}
By hypothesis, $g_{\S }\mathcal{D}_{i}=0$ and hence, $\left(f_{\S }\otimes \S \right)\mathbf{D}_{i}\rho=\left(f_{\S }\otimes \S \right)\rho \mathcal{D}_{i}$. The fact that $\rho$ commutes with each $\sigma\in\mathcal{\S }_{N}$ implies that $\left(f_{\S }\otimes \S \right)\mathbf{U}_{i}\rho=\left(f_{\S }\otimes \S \right)\rho\mathcal{U}_{i}$, for each $i$.
\end{proof}
Recently, in \cite{GGJL17}, the authors study this projection via its interpretation as a homomorphism of standard modules of the rational Cherednik algebra.

Now, we present the analogous results for the Hecke algebra case.
\begin{prop}\label{singPhi}
A polynomial $p$ is singular if and only if $p\xi_{i}=p\phi_{i}$ for $1\leq i\leq N$, where $\phi_i$ are the Jucys-Murphy elements in $\mathcal{H}_N(t)$.
\end{prop}

\begin{proof}
We proceed by induction. 

By definition, $p\mathcal{D}_{N}=0$ if and only if $p\xi_{N}=p=p\phi_{N}$.

Now, suppose that $p\mathcal{D}_{j}=0$, for $i<j\leq N$, if
and only if $p\xi_{j}=p\phi_{j}$, for $i<j\leq N$. Then, $\displaystyle{0=p\mathcal{D}_{i}=\frac{1}{t}pT_{i}\mathcal{D}_{i+1}T_{i}}$ if and only if $pT_{i}\mathcal{D}_{i+1}=0$, which in terms of the Jucys-Murphy elements, is equivalent to say that $ pT_{i}\xi_{i+1}=pT_{i}\phi_{i+1}$. Multiplying both sides on the right by $T_i$ and applying the inductive hypothesis, we get that $0=p\mathcal{D}_{i}$ if and only if $tp\xi_{i}=tp\phi_{i}$, and this completes the proof.
\end{proof}

The following result gives us an idea of how useful the property (SMP) can be. 
\begin{thm}\label{Hrhocomm}
Suppose $\left\{g_{\S }:\S \in\tab_\tau\right\}$ satisfies property (SMP). Then, the map $\rho$ defined in Proposition \ref{Hcomm} intertwines $\mathbf w^q$, $\mathbf E_i$ and $\mathbf{D}_{i}^{q,t}$ with $\omega^q$, $\xi_i$ and $\mathcal{D}_i^{q,t}$ respectively, for $1\leq i \leq N$. That is, $\mathbf w^q \rho =\rho \omega^q$, $\mathbf E_i \rho = \rho \xi_i$ and $\mathbf D_i^{q,t} \rho = \rho \mathcal{D}_i^{q,t}$, for $1\leq i \leq N$.
\end{thm}

\begin{proof}
We analyze the commutation with $\mathbf w^q$ since the other commutations follow from the definitions of $\xi_i,\mathbf E_i, \mathcal{D}_i^{q,t}$ and $\mathbf{D}_i^{q,t}$. 
For $\mathbf w^q$, we have that
\begin{multline*}
\left(f_\S (x)\otimes \S\right) \mathbf w^q \rho = \left[ t^{1-N}\left(f_\S(x)\omega^q\right)\otimes \S\tau(T_1\dots T_{N-1})\right]\rho= \\  
=\left[f_\S(x)\omega^q \otimes t^{1-N}\S\tau(T_1\dots T_{N-1}) \right]\rho  = \left( f_\S(x)\omega^q\right) \left(t^{1-N}g_\S T_1\dots T_{N-1}\right) = \\ = \left( f_\S(x)\omega^q\right)\left( g_\S(x)\omega^q\right) = \left(f_\S(x)g_\S(x)\right) \omega^q = \left(f_\S(x)\otimes \S\right)\rho \omega^q, 
\end{multline*}
where we use that $t^{1-N}g_\S T_1\dots T_{N-1} = g_\S \omega^q$. To see this, consider the Jucys-Murphy element $\phi_1$. By Proposition \ref{singPhi}, $g_\S \xi_1 = g_\S\phi_1$ and then,
\begin{eqnarray*}
g_\S\xi_1 = g_\S\omega^q T_{N-1}\dots T_1 = g_\S\phi_1 = t^{1-N}g_\S T_1\dots T_{N-1}T_{N-1}\dots T_1. 
\end{eqnarray*}
Cancelling out the factor $T_{N-1}\dots T_1$ gives the desired result. 
\end{proof}

\begin{remark}
 Using a similar argument to Lemma \ref{LemmaJM}, we show that
\begin{multline*}
\sum_{i=1}^{N}x_{i}\left(  g\mathcal{D}_{i}\right)    =\sum_{i=1}^{N}x_{i}\frac{\partial g}{\partial x_{i}}+\kappa\sum\limits_{i<j}\left(g\left(  x\right)  -g\left(  xs_{ij}\right)  \right) = \\
 =\sum_{i=1}^{N}x_{i}\frac{\partial g}{\partial x_{i}}+\kappa\left(\frac{N\left(  N-1\right)  }{2}g\left(  x\right)  -\sum_{i=1}^{N}g\left(x\right)  \omega_{i}\right) .
\end{multline*}
Thus, if $g$ is homogeneous of degree $n$, singular for $\kappa=\kappa_{0}$ and of isotype $\tau$, then these parameters satisfy that $n+\kappa_{0}\left(\frac{N\left(N-1\right)}{2}-\Sigma\left(\tau\right)\right)=0$, where $\Sigma\left(\tau\right)$ is the sum of the contents in the Ferrers diagram of $\tau$ (i.e. $\Sigma\left(\tau\right)=\sum_{\left(i,j\right)  \in\tau}\left(j-i\right)$).

The role of this formula is more relevant for Macdonald polynomials indexed by quasistaircase partitions as we describe later in the subsection \ref{SubsubFinalRemark}. 
\end{remark}

\section{Jack and Macdonald polynomials\label{sec:JackMacdo}}
There are several types of Jack and Macdonald polynomials; for instance, symmetric and nonsymmetric, homogeneous and nonhomogeneous (so called shifted) etc. We are interested in the symmetric and nonsymmetric Jack and Macdonald polynomials, and in their vector-valued versions. In this section we present the ones that are relevant for our study. For more details see \cite{DL2012}.

\subsection{Jack polynomials} 
Consider a formal parameter $\kappa$. There is a basis of simultaneous eigenfunctions for the Cherednik operators, called \emph{nonsymmetric Jack polynomials} $\{J_\alpha(x)\}$, with $\vartriangleright$-leading term $x^{\alpha}$ for $\alpha\in\comps$. 
The eigenvectors $\varsigma_\alpha$, called \emph{spectral vectors}, are given by $\varsigma_\alpha(i)=\alpha_i+1+\kappa(N-r(\alpha,i))$, where $r(\alpha,i)$ is the rank function defined by \begin{eqnarray*}
r(\alpha,i):= \#\{ j:\ \alpha_j>\alpha_i\}+\# \{ j:\ 1\leq j\leq i,\ \alpha_j=\alpha_i\}. 
\end{eqnarray*}
Observe that the function $i\mapsto r(\alpha,i)$ is a permutation of $\{1,2,\dots,N\}$ and that $\alpha$ is a partition if and only if $r(\alpha,i)=i$, for all $i$. 

Using the commutation relations of $s_i$ and $\mathcal{U}_i$, it can be shown that
\begin{eqnarray*}
\left\{ \begin{array}{l}
J_{\alpha}s_i=J_\alpha, \hspace{0.5cm} \text{ if } \alpha_i=\alpha_{i+1}, \\[0.05in]
\displaystyle{J_\alpha s_i=J_{\alpha s_i} + \frac{\kappa}{\varsigma_\alpha(i) - \varsigma_\alpha (i+1)} J_\alpha}, \hspace{0.5cm} \text{ if }\alpha_i<\alpha_{i+1}.
\end{array}\right.
\end{eqnarray*}

The \emph{symmetric Jack polynomials}, $\mathrm J_{\lambda}$, are obtained by applying symmetrizing operators to the nonsymmetric Jack polynomial $J_{\alpha}$ such that $\alpha^+=\lambda$. In terms of operators, the symmetric Jack polynomials are the eigenvectors of $\sum_i \mathcal{U}_i$. Moreover, for a partition $\lambda$ and $m\geq 1$, 
\begin{eqnarray*}
\mathrm J_\lambda \sum_{i=1}^N \mathcal{U}_i^m = \sum_{i=1}^N \left[ \lambda_i +1+\kappa (N-i+1) \right]^m \mathrm J_\lambda.
\end{eqnarray*}
To the question of how many of these sums suffice to separate different partitions, the answer is that considering $m=1,2,\dots, N$ is certainly enough. Notice also that considering only $m=1$ does not work. 

Finally, there are simultaneous eigenfunctions of $\{\mathbf U_{i}\}$, called \emph{vector-valued nonsymmetric Jack polynomials}, $J_{\alpha, \S}$, with $(\alpha,\S ) \in \comps \times \tab_\tau$. The eigenvectors are given by the spectral vectors $\varsigma_{\alpha,\S }(i)=\alpha_i +1+\kappa\cdot \ct{S}{r(\alpha,i)}$. For more details, see \cite{DL2011}.

\subsection{Macdonald polynomials}
Consider two formal parameters $q$ and $t$. There is a basis of simultaneous eigenfunctions for the Cherednik-Dunkl operators, the \emph{nonsymmetric Macdonald polynomials} $M_{\alpha}$, labeled by $\alpha\in\comps$, with $\vartriangleright$-leading term $t^{\ast}x^{\alpha}$. For us, $\ast$ denotes the exponent of $t$, which is not relevant for the order on the monomials $x^\alpha$. In this case, the spectral vector is $\zeta_{\alpha}(i)=q^{\alpha_{i}}t^{N-r(\alpha,i)}$. Notice that from the definition of the Cherednik-Dunkl operators, the Macdonald polynomials $M_{\alpha}$ are necessarily homogeneous. 

As for Jack polynomials, Macdonald polynomials satisfy the following recursive formulas:
\begin{eqnarray*}
\left\{\begin{array}{l}
M_{\alpha}T_{i}=tM_{\alpha}, \hspace{0.5cm} \text{ if } \alpha_{i}=\alpha_{i+1},\\[0.05in]
\displaystyle{M_{\alpha}T_{i}=M_{\alpha s_{i}}-\frac{1-t}{1-\zeta_{\alpha}(i+1)/\zeta_{\alpha}(i)}M_{\alpha}},\hspace{0.5cm} \text{ if } \alpha_{i}<\alpha_{i+1}.
\end{array}\right.
\end{eqnarray*}

The \emph{symmetric Macdonald polynomials} $\mathrm M_{\lambda}$ are defined as the eigenvectors of $\sum_{i}\xi_{i}$ and are indexed by partitions $\lambda$. 

In this case, by the eigenfunction properties of symmetric Macdonal polynomials, for a partition $\lambda$ and $m\geq 1$,
\begin{eqnarray*}
\mathrm M_\lambda \sum_{i=1}^N \xi_i^m = \sum_{i=1}^N \left[q^{\lambda_i}t^{N-i}\right]^m \mathrm M_\lambda
\end{eqnarray*}
As a remarkable difference between Jack and Macdonald polynomials, the sum for $m=1$ alone works for generic parameters $q$ and $t$ and the $\lambda_i$ is determined by the exponent of $q$. This is an example of the meta-principle that adding a parameter can simplify the problem. 
 
The \emph{vector-valued nonsymmetric Macdonald polynomials}, $M_{\alpha,\S}$, are defined as the simultaneous eigenfunctions of the operators $\{\mathbf E_{i}\}$ in $\mathcal{P}_\tau$.

\subsection{Projections and highest weight}\label{SubSectProjSingHW}

As we mentioned before, the notion of singularity has a symmetric counterpart. Let us start with the definition and leave the motivation for the beginning of the next section. 
\begin{defi}
A polynomial is \emph{highest weight} if it belongs to the kernel of the sum of the Dunkl operators $\sum_i \mathcal{D}_i$, for the symmetric case, or to the kernel of $\sum_i \mathcal{D}_i^{q,t}$, for the Hecke algebra case. 
\end{defi}
The following two corollaries show that the projections defined in Section \ref{subsect:proj} in a more general setting restrict nicely to the Jack and Macdonald polynomials. 
\begin{cor}
Under the same hypotheses as in Proposition \ref{SPcomm}, $\rho$ maps each vector-valued nonsymmetric Jack polynomial to a scalar simultaneous $\left\{\mathcal{U}_{i}\right\}$-eigenfunction, which is a Jack polynomial if the spectral vectors $\zeta_{\alpha}$ for the corresponding degree determine $\alpha$ uniquely. Otherwise, it maps to zero. 
\end{cor}
\begin{remark}
In general, there might be some negative rational $\kappa$ values for which the spectral vectors $\zeta_\alpha$ do not determine $\alpha$ uniquely. 
\end{remark}

\begin{cor}
Under the same hypotheses as in Proposition \ref{Hrhocomm}, $\rho$ maps each vector-valued nonsymmetric Macdonald polynomial to a scalar simultaneous $\left\{\xi_{i}\right\}$-eigenfunction, which is a Macdonald polynomial if the spectral vectors $\zeta_{\alpha}$ for the corresponding degree determine $\alpha$ uniquely. Otherwise, it maps to zero. 
\end{cor}

\section{The quasistaircase partition}\label{SectionQSP}
\subsection{The importance of the quasistaircases}
The investigation on highest weight Jack polynomial is motivated by the quest for the
description of the wave functions that model the fractional quantum Hall effect (FQH). The first occurrence
of a Jack polynomial (at that time we did not know it was one yet) is the Laughlin wavefunction itself \cite{Laughlin}:
\begin{eqnarray*}
	\varphi^{m}_{L}(z_{1},\dots,z_{N})=\prod_{i<j}\left(z_{i}-z_{j}\right)^{2m}.
\end{eqnarray*}
This function depends on an integer parameter $m$ and is known to be a \emph{true} staircase Jack polynomial for the parameter $\alpha=-\frac1m$.{}
The study of $\varphi^{m}_{L}$ as a symmetric function generated literature for the purpose of the expansion in the Schur basis and the monomial basis \cite{dFGIL,TSW}. This first case is particularly interesting because it is also related to rectangular Jack polynomials and hyperdeterminants \cite{BBL,BLT,Luque}. Notice that, even in this (simplest) case, Macdonald polynomial provides a more regular framework for the study of these functions. For instance, partitions involved in the expansion in the Schur basis are more easily understood for the $q$-deformation \cite{BL,KTW}.

Although the description of this wave function is not known in all the FQHE configurations, this particularly simple expression led physicists to note that it resulted from the coincidence of two notions: clustering properties and highest weight polynomials. The clustering properties was in particular studied by Feigin \emph{et al.} in the early 2000s \cite{FJMM1,FJMM2} and described in terms of ideal, \emph{wheel} conditions and staircase partitions. The second notion (highest weight) is related to some differential operators (or their $q$-deformations). A family of highest weight Macdonald polynomials was described in \cite{JL}. The part of the results concerning staircases can be recovered from the works of Feigin \emph{et al.} while, for the polynomials indexed by quasistaircases which are not staircases, it is a non trivial consequence of the Lassalle binomial formulas \cite{Lassalle}. 

Several Jack polynomials are used to describe some FQH states (or wave functions thought to be adiabatically related to the true eigenstates) \cite{BH2008,BH2,MR,RR}. Nevertheless, it is not always the case and some FQH states seem not to be directly related to Jack polynomials. The hope of completing the picture with Macdonald polynomials comes from the theory of nonsymmetric Jack and Macdonald polynomials. These polynomials can be considered as elementary building blocks allowing to reconstitute symmetric polynomials by linear combinations. There is a double challenge with this approach. Firstly, we have to reconstruct the clustering and the highest weight notions from properties of nonsymmetric polynomials. Secondly, we have to write FQH wave functions nicely in terms of these elementary building blocks. In the present paper we deal with the first problem. Let us be more precise on this point.
The clustering properties comes from physical interpretation of constraints on the position of the particles. These constraints translate in terms of polynomials as vanishing properties. More precisely, the polynomial representing the wave function vanishes when $k$ particles tends to the same position. The number $k$ is one of the parameters describing the clustering property. A second parameters take into account the number of clusters which are required to observe the vanishing property and another parameter allows to take into account the strength with which particles repel each other. In terms of polynomial, the more strongly the particles repel each other, the faster the polynomial tends to $0$ when $k$ variables tend to the same value. For a mathematical purpose, instead of stating the
results in terms of clustering properties, we prefer an equivalent statement in terms of factorizations \cite{CDL2017}.
The computation of the exact wavefunctions being a very difficult problem, physicists  overcome this difficulty by empirically research for polynomials that are adiabatically related to true eigenvalues and are guided by the clustering properties.
In this  quest of wavefunctions, Bernevig and Haldane \cite{BH2008,BH2}, based on results from Feigin et al \cite{FJMM1}, noted that some of Jack's polynomials indexed by quasistaircase partitions could be good candidates. In that context they gave three conjectures that we translate below in terms of factorization. We assume that the parameters $k+1$ and $s-1$ are coprime and we set $\lambda^{\beta}_{k,r,s}=[(\beta r+s(r-1)+1)^{k},\dots,(s(r-1)+1)^{k},0^{n_{0}}]$ with $n_{0}=(k+1)s-1$ and the number of particles (\emph{i.e.} the size of the partition) is $N=\beta k+n_{0}$. Notice that $\beta$ is not really a parameter since it depends on $N$ and the three other parameters. The three conjectures read:
\begin{enumerate}
	 \item {\bf First clustering property} They considered $s-1$ clusters of $k+1$ particles 
	 $Z_{1}=z_{1}=\cdots=z_{k+1}$, $Z_{2}=z_{k+1}=\cdots=z_{2(k+1)}$,$\dots$,$Z_{s-1}=z_{(s-2)(k+1)+1}=\cdots=z_{(s-1)(k+1)}$, together with 
	 $k$ particle cluster $Z_{F}=z_{(s-1)(k+1)}=\cdots=z_{s(k+1)-1}$.
	  The other particles (variables) remain free. For such a specialization,
	 the Jack polynomial $J^{-{k+1\over r-1}}_{\lambda_{k,r,s}^{\beta}}((k+1)(Z_{1}+\cdots+Z_{s-1})+kZ_{F}+z_{s(k+1)}+\cdots+z_{N})$ behaves as $\displaystyle\prod_{i=s(k+1)^{N}}\left(Z_{F}-z_{i}\right)^{r}$ when 
	 each $z_{i}$ ($i=s(k+1),\dots,N$) tends to $Z_{F}$ but has a remaining polynomial factor which does not tends to $0$.\\ For instance,
	 $
	 J^{(-2)}_{53}(2Z_{1}+Z_{F}+z_{3}+z_{4})= \left( {\it Z_F}-{\it z_4}
 \right) ^{2} \left( {\it Z_F}-{\it z_3} \right) ^{2}P(Z_{1},Z_{F},z_{3},z_{4})$ where $P$ is a degree $4$ polynomial.
	 \item {\bf Second clustering property} They considered a  cluster of $n_{0}=(k+1)s-1$ particles $z_{1}=\cdots=z_{(k+1)s-1}=Z$.
	 The  Jack polynomial $J_{\lambda^{\beta}_{k,r,s}}^{\left(-{k+1\over r-1}\right)}(n_{0}Z+z_{n_{0}+1}+\cdots+z_{N})$ behaves as
	 $\displaystyle\prod_{i=s(k+1)}^{N}(Z-z_{i})^{(r-1)s+1}$ when each $z_{i}$ tends to $Z$. More specifically, for highest weight Jack polynomials, one has the 
	 following explicit formula:
\begin{eqnarray*}
	J_{\lambda_{k,r,s}^{\beta}}^{\left(-{k+1\over r-1}\right)}(n_{0}Z+z_{n_{0}+1}+\cdots+z_{N})\displaystyle{\mathop=^{(*)}}
\prod_{i=s(k+1)}^{N}(Z-z_{i})^{(r-1)s+1}
	J_{\lambda_{k,r,1}^{\beta-1}}^{\left(-{k+1\over r-1}\right)}(z_{n_{0}+1}+\cdots+z_{N}).
\end{eqnarray*}
	 \item {\bf Third clustering property} It is obtained by forming  $s-1$ clusters of $2k+1$ particles $Z_{1}=z_{1}=\cdots=z_{2k+1}$,$\dots$,{}
	 $Z_{s-1}=z_{(s-2)(2k+1)+1}=\cdots=z_{(s-1)(2k+1)}$. A highest weight Jack $J_{\lambda^{\beta}_{k,r,s}}$ satisfies
	 \begin{multline*}
\displaystyle J_{\lambda^{\beta}_{k,r,s}}^{\left(-{k+1\over r-1}\right)}((2k+1)(Z_{1}+\cdots+Z_{s-1})+z_{(s-1)(2k+1)+1}+\cdots+z_{N}) \\
\displaystyle{\mathop=^{(*)}}
\prod_{1\leq i<j\leq s-1}(Z_{i}-Z_{j})^{k(3r-2)}\displaystyle\prod_{i=1}^{s-1}\prod_{\ell=(s-1)(2k+1)+1}^{N}(Z_{i}-z_{\ell})^{2r-1} \\
\times J_{\lambda^{\beta-s+1}_{k,r}}(z_{(s-1)(2k+1)+1}+\cdots+z_{N}).
	 \end{multline*}
 \end{enumerate}

Notice that authors proved the second conjecture \cite{CDL2017} and, although the proof is certainly more technical, we guess that the third conjecture should be proved following a similar method. The first conjecture seems more difficult because it involves a factor that has not been identified. Our strategy has been to extend these results in three directions :
\begin{enumerate}
	\item To nonsymmetric polynomials in order to use the inductions involved in the Yang-Baxter graph; 
	\item To nonhomogeneous polynomials, because the vanishing properties are  simpler to study in that context;
	\item And to Macdonald polynomials in order to \emph{polarize} the vanishing properties in the sense that, in the factorization formulas, the linear factors are pairwise distinct; this is easier to check algebraically.
	\end{enumerate}

 Nonsymmetric Jack polynomials were introduced by Opdam \cite{Opdam} and were used to describe the polynomial part of the eigenfunctions of the Calogero--Sutherland model on a circle and define families of orthogonal polynomials which are multivariate analogues of Hermite and Laguerre polynomials \cite{BF}. Nonsymmetric Macdonald polynomials were introduced in \cite{Cherednik, Macdonald}. Marshall \cite{Marshall} tied the theories of symmetric and nonsymmetric Macdonald polynomials by adapting the approach of \cite{BF} to the $(q,t)$-deformation.

The nonsymmetric counterpart of the highest weight is the notion of singularity. As we define previously, a polynomial is singular if it belongs to the kernels of the Dunkl operators \cite{Dunkl1}. In \cite{D2005} one of the authors shows that quasistaircase Jack polynomials, under some conditions, are singular. The notion of isotype for Jack polynomials used in the current paper comes from this article. One of our goals in this paper is to adapt the results and proofs developed in \cite{D2005} to the Macdonald case.

Clustering properties of highest weight Macdonald polynomials are consequences of their factorizations into linear factors when submitted to some specializations. The main tool to show the factorization is another kind of Macdonald polynomials: the shifted Macdonald polynomials are the nonhomogeneous version of the symmetric Macdonald polynomials and have the additional property of being defined alternatively by vanishing conditions (see e.g. \cite{Lascoux1,Lascoux2}). The factorizations are obtained by combining homogeneity of highest weight shifted Macdonald polynomials together with their vanishing properties. These factorizations were investigated by the authors in two previous papers: for rectangular partitions in \cite{DL2015} and for general quasistaircases in \cite{CDL2017}.
Numerical evidences suggest that the clustering properties have nonsymmetric analogues. We gave several examples at the end of \cite{CDL2017} that involve quasistaircase partitions. The
purpose of the present paper is to establish links between the notion of highest weight polynomials and singular polynomials. The problem of proving clustering properties of quasistaircase nonsymmetric Macdonald polynomials will be done in a future work.

\subsection{The quasistaircase partition}
The most general \emph{quasistaircase partition}, or simply quasistaircase, is defined as follows. 
\begin{defi}
A quasistaircase partition is a partition of the form
\begin{eqnarray*}
qs(m,n,p,K,s,\nu):=\left[((K-1)m+p)^{\nu},((K-2)m+p)^{n},\dots,p^{n},0^{s}\right],
\end{eqnarray*}
with $1\leq \nu\leq n-1$ and where $a^l$ means that the part $a$ occurs $l$ times. 
\end{defi}
For instance, $qs(5,3,2,4,7,2)=\left[17^2,12^3,7^3,2^3,0^7 \right]$.

Instead of staying with this general definition, we look at a specific description of the quasistaircase. Consider a pair $(m,n)\in\mathbb{N}_{0}^{2}$, with $2\leq n\leq N$
and $\frac{m}{n}\notin\mathbb{N}_{0}$. Moreover, let $d=\gcd\left(m,n\right)$, $\displaystyle{m_{0}=\frac{m}{d}}$, $\displaystyle{n_{0}=\frac{n}{d}}$, $\displaystyle{\kappa_{0}=-\frac{m_{0}}{n_{0}}=-\frac{m}{n}}$, and
$\displaystyle{l=\left\lceil \frac{N-n+1}{n_{0}-1}\right\rceil +1}$. We consider a partition $\tau=\left[n-1,(n_{0}-1)^{l-2},\tau_{l}\right]$ of $N$ with $\tau_{l}\leq n_{0}-2$. Then, we consider quasistaircase partitions of the form $\mu=qs(m_{0},n_{0}-1,m,l-1,n-1,\tau_{l})$. 

For instance, consider $N=10$, $n=6$, $m=4$, and $\tau=(5,2,2,1)$. Then, $n_0=3$, $m_0=2$, $\kappa_0=\frac{-2}{3}$ and $\mu =\left[8,6,6,4,4,0^5\right]$.

Although we have introduced more parameters in our description of the quasistaircase, they are advantageous for the specialization of the parameters $(q,t)$ and the study of the Macdonald and Jack polynomials. Furthermore, from now on, $\mu$ denotes the quasistaircase partition $\mu=qs(m_{0},n_{0}-1,m,l-1,n-1,\tau_{l})$.

\subsection{The symmetric group case}
In this subsection we summarize the work already done for the symmetric group case in order to provide the results that we use as a guide for the Hecke algebra case. We also give a brief overview of those proofs that are relevant. 
\subsubsection{Singularity of quasistaircase nonsymmetric Jack polynomials}
The following proposition summarizes the results obtained in \cite{D2005} related to the symmetric group case.
 \begin{prop}[\cite{D2005}]
There exists an irreducible $\mathcal{S}_{N}$-module of isotype $\tau$ with Ferrers diagram $\mu=qs(m_{0},n_{0}-1,m,l-1,n-1,\tau_{l})$, and constituted with polynomials that are singular for $\kappa=\kappa_{0}$. 
Moreover, for each $\S \in\tab_\tau$, $\alpha(\S )^+= \mu$ and the set $\left\{g_{\S }=J_{\alpha\left(\S \right)}:\ \S \in\tab_\tau\right\}$ satisfies property (SP), where the nonsymmetric Jack polynomials are specialized to $\kappa=\kappa_{0}$.
\end{prop}

The following result describes the relation between $\S \in \tab_\tau$, $\alpha(\S )$, and $\kappa_0$.
\begin{prop}\label{SPeigv}
Consider $\S \in\tab_\tau$ and $k$ such that $1\leq k\leq N$. Then,
$$\alpha(\S )_{k}+\kappa_{0}\left(N-r(\alpha(\S ),k)\right)=\kappa_{0}\ct{S}{k}.$$
\end{prop}

\begin{proof}
Fix $k$ and $\S \in\tab_\tau$, and let $\row{S}{k}=i$ and $\col{S}{k}=j$. Every entry $u$ in the $i^{\text{th}}$ row with higher column index (to the right) has $\alpha(\S )_u=\alpha(\S )_k$ with $u<k$. Also, every entry $u$ in the $p^{\text{th}}$ row, with $p>i$, has $\alpha(\S )_u>\alpha(\S )_k$. Thus, $\displaystyle{r(\alpha(\S ),k)=\tau_i-j+1+\sum_{p=i+1}^{\ell(\tau)}\tau_p}$, and $\displaystyle{N-r(\alpha(\S ),k)= j-1-\sum_{p=1}^{i-1} \tau_p}$. 
 
We analyze $\alpha(\S )_k$ depending on $i$. 
If $i=1$, then $\alpha(\S)_{k}=0$ and $N-r\left(\alpha(\S ),k\right) =j-1=\ct{S}{k}$.

Suppose now that $\tau$ has two rows and $i=2$, then $N-r\left(\alpha(\S),k\right)=n-1+j-1$ and $\alpha(\S )_{k}=m$. Therefore,
$$\alpha(\S )_{k}-\frac{m}{n}\left(N-r\left(\alpha(\S )k\right)\right)=\frac{m}{n}\left(2-n-j+n\right)=-\frac{m}{n}\ct{S}{k}.
$$
Finally, suppose $\tau$ has 3 or more rows and $i\geq2$. By construction, $\alpha(\S)_{k}=m+(i-2)m_{0}$. Moreover, $N-r\left(\alpha(\S ),k\right)=j-1+(n-1)+\left(i-2\right)\left(n_{0}-1\right)$. Then, 
\begin{multline*}
\alpha(\S )_{k}-\frac{m_{0}}{n_{0}}\left[N-r\left(\alpha\left(\S \right),k\right)\right] =\frac{m_{0}}{n_{0}}\left[dn_{0}+\left(i-2\right)n_{0}-n-j+2-\left(i-2\right)\left(n_{0}-1\right)\right]= \\=\frac{m_{0}}{n_{0}}\left(i-j\right)=\kappa_{0}\ct{S}{k}.
\end{multline*}
\end{proof}

Note also that if some $J_{\alpha}$ is singular, then the $\mathcal{S}_{N}$-invariant subspace generated by the set $\left\{J_{\alpha}\sigma:\ \sigma\in\mathcal{S}_{N}\right\}$ consists of singular polynomials, because $s_{ij}\mathcal{D}_{i}s_{ij}=\mathcal{D}_{j}$, for all $i\neq j$.

The following statements apply to the singular Jack polynomials $J_{\alpha(\S )}$.
\begin{thm}[{{\cite[Theorem 5.7]{D2005}}}]\label{Jsing}
The polynomials $J_{\alpha(\S )}$ with $\kappa=-\frac{m}{n}$ satisfy property (SP).
\end{thm} 
\begin{proof}
The first step is to prove that $J_{\mu}$ is of isotype $\tau$ by means of a representation-theoretic argument about inducing the trivial representation of $\mathcal{S}_{\mu}$ (the stabilizer of $\mu$) up to $\mathcal{S}_{N}$ \cite[Theorem 5.2, Proposition 5.3]{D2005}\cite{Macd1995}. This implies that any polynomial derived from $J_{\mu}$ by group action is of the same isotype.

It is clear that if $\row{S}{i}=\row{S}{i+1}$, then $\alpha(\S)_{i}=\alpha(\S)_{i+1}$ and $J_{\alpha}s_{i}=J_{\alpha}$. Suppose $\alpha(\S)_{i}<\alpha(\S)_{i+1}$, which is equivalent to $\row{S}{i}<\row{S}{i+1}$. Then, $\S ^{(i)}\in\tab_\tau$, and
\begin{eqnarray*}
\S \tau\left( s_{i}\right) =\S^{(i)}+\frac{1}{\ct{S}{i} - \ct{S}{i+1}}\S .
\end{eqnarray*}

The spectral vector satisfies $\zeta_{\alpha(\S)}(i)=\alpha(\S)_{i}+1+\kappa_{0}\left( N-r\left(\alpha(\S) ,i\right)\right)=1+\kappa_{0}\ct{S}{i}$, thus
\begin{eqnarray*}
J_{\alpha\left( \S \right)}s_{i}=J_{\alpha\left(\S \right)s_{i}}+\frac{\kappa_{0}}{\kappa_{0}\ct{S}{i}-\kappa_{0}\ct{S}{i+1}}J_{\alpha\left(\S \right)}.
\end{eqnarray*}
Notice that $\alpha\left(\S^{(i)}\right) =\alpha(\S)s_{i}$ by definition. It remains to show that $J_{\alpha(\S)}s_{i}=-J_{\alpha(\S)}$ when $\col{S}{i}=\col{S}{i+1}$. By the action of the Jucys-Murphy elements, $J_{\alpha(\S)}\omega_{i}=\ct{S}{i}J_{\alpha(\S)}$. Consider the polynomial $f=J_{\alpha(\S)}+J_{\alpha(\S)}s_{i}$ and the relation $s_{i}\omega_{i}s_{i}=\omega_{i+1}+s_{i}$; then 
\begin{eqnarray*}
f\omega_{i}=\ct{S}{i} J_{\alpha(\S)}+\ct{S}{i+1} J_{\alpha(\S) }s_{i}+J_{\alpha(\S)}=\ct{S}{i+1} f,
\end{eqnarray*}
and similarly $f\omega_{i+1}=\ct{S}{i}f$. Also $f\omega_{j}=\ct{S}{j}f$, for $j<i$ or $j>i+1$ (by the commutation $s_{i}\omega_{j}=\omega_{j}s_{i}$). Then $f$ is of isotype $\tau$ but has impossible eigenvalues for each $\omega_{j}$, thus $f=0$.
\end{proof}
\begin{remark}
Historically, the method of constructing singular polynomials was to show directly that $J_{\mu}$ is singular for certain partitions $\lambda$ and then to analyze the $\mathcal{S}_{N}$-invariant subspace generated by $J_{\mu}$. It turned out that this subspace corresponds to an irreducible representation of $\tau$ and this led to the correspondence between $J_{\alpha}$ and RSYTs of shape $\tau$, where $\alpha$ is a reverse lattice permutation of $\mu$. 

Part of the proof of this fact concerns the property $J_{\alpha(\S)}s_{i}=-J_{\alpha(\S)}$ when $\col{S}{i}=\col{S}{i+1}$, and we sketch it in the proof of Theorem \ref{Jsing}. More details of the argument establishing the $\tau$-isotype property can be found in the proof of Theorem \ref{Hsingp}. In recent work, currently a paper under preparation, there is a direct proof of the isotype property which is used to prove singularity by means of the Jucys-Murphy elements in Lemma \ref{LemmaJM} and in Proposition \ref{singPhi}.
\end{remark}

Next we look at the implications for $\rho.$ We describe sets of Jack polynomials $\left\{J_{\beta,\S }\right\}$ which span irreducible $\mathcal{S}_{N}$-invariant subspaces of $\mathcal{P}_{\tau}$. We need another definition first.
\begin{defi}
For $\alpha\in\comps$ and $\S \in\tab_\tau$ define $\left\lfloor \alpha,\S \right\rfloor $ to be the filling of the Ferrers diagram of $\tau$ obtained by replacing $i$ by $\alpha_{i}^{+}$ in $\S $, for all $i$. 
This does not define a one to one correspondence, and so we consider the following set $\mathcal{T}(\alpha,\S)=\left\{\left(\beta,\S^{\prime}\right):\left\lfloor \beta,\S^{\prime}\right\rfloor=\left\lfloor \alpha,\S \right\rfloor \right\}$.
\end{defi}
For instance, take $\tau=(3,2)$, $\alpha=(1,4,2,0,3)$ and 
$\S =\scalebox{0.7}{\begin{tikzpicture}
\draw (0,0) rectangle (0.5,0.5);
\draw (0.5,0) rectangle (1,0.5);
\draw (1,0) rectangle (1.5,0.5);
\draw (0,0.5) rectangle (0.5,1);
\draw (0.5,0.5) rectangle (1,1);
\node at (0.25,0.25) {5};
\node at (0.75,0.25) {4};
\node at (1.25,0.25) {1};
\node at (0.25,0.75) {3};
\node at (0.75,0.75) {2};
\end{tikzpicture}}$. Then,
$\left\lfloor \alpha,\S \right\rfloor =
\scalebox{0.7}{\begin{tikzpicture}
\draw (0,0) rectangle (0.5,0.5);
\draw (0.5,0) rectangle (1,0.5);
\draw (1,0) rectangle (1.5,0.5);
\draw (0,0.5) rectangle (0.5,1);
\draw (0.5,0.5) rectangle (1,1);
\node at (0.25,0.25) {0};
\node at (0.75,0.25) {1};
\node at (1.25,0.25) {4};
\node at (0.25,0.75) {2};
\node at (0.75,0.75) {3};
\end{tikzpicture}}\ $.

The definition of $\mathcal{T}(\alpha,\S)$ is motivated by the following result. 
\begin{prop}[\cite{DL2011,DL2012}]
$J_{\alpha,\S }$ can be transformed to $J_{\beta ,\S^{\prime}}$ by a sequence of maps of the form $f\rightarrow af+bfs_{i}$ if and only if $\left\lfloor \alpha,\S \right\rfloor =\left\lfloor \beta,\S^{\prime}\right\rfloor $.
Furthermore, it is known that in this case, the spectral vectors of $\left(\alpha,\S \right)$ and $\left(\beta,\S ^{\prime}\right)$ are permutations of each other.
\end{prop}
Next result identifies when there is a unique nonzero symmetric polynomial in terms of $\left\lfloor \alpha,\S \right\rfloor$.
\begin{thm}[\cite{DL2011}]
For $\left(\alpha,\S \right)\in\mathbb{N}_{0}^{N}\times\tab_\tau$, the $\mathrm{span}\left\{J_{\beta,\S ^{\prime}}:\left(\beta,\S ^{\prime}\right)\in\mathcal{T}\left(\alpha,\S \right)\right\}$ contains a unique nonzero symmetric polynomial if and only if $\left\lfloor \alpha,\S \right\rfloor $ is a column-strict tableau.
\end{thm}
For each shape $\tau$, consider the unique minimal column-strict tableau of shape $\tau$ obtained by filling the $i^{\text{th}}$ row with the entries $i-1$ for each $i$. Then, the partition $\lambda=\left(\left(l-1\right)^{\tau_{l}},\left(l-2\right)^{\tau_{l-1}},\ldots,1^{\tau_{2}},0^{\tau_{1}}\right)$ and the inv-minimal RSYT $\S _{1}$ is the unique $\S$ such that $\left\lfloor \lambda ,\S\right\rfloor $ equals this tableau. 
For example, take $\tau=\left(5,2,2,1\right)$. Then, $\lambda=\left(3,2,2,1,1,0^{5}\right)$, 
\begin{eqnarray*}
\S _{1}=
\scalebox{0.7}{
 \begin{tikzpicture}
\draw (0,0) rectangle (0.5,0.5);
\draw (0.5,0) rectangle (1,0.5);
\draw (1,0) rectangle (1.5,0.5);
\draw (1.5,0) rectangle (2,0.5);
\draw (2,0) rectangle (2.5,0.5);
\draw (0,0.5) rectangle (0.5,1);
\draw (0.5,0.5) rectangle (1,1);
\draw (0,1) rectangle (0.5,1.5);
\draw (0.5,1) rectangle (1,1.5);
\draw (0,1.5) rectangle (0.5,2);
%\draw (1.5,0) rectangle (2,0.5);
\node at (0.25,0.25) {10};
\node at (0.75,0.25) {9};
\node at (1.25,0.25) {8};
\node at (1.75,0.25) {7};
\node at (0.25,0.75) {5};
\node at (2.25,0.25) {6};
\node at (0.75,0.75) {4};
\node at (0.25,1.25) {3};
\node at (0.75,1.25) {2};
\node at (0.25,1.75) {1};
\end{tikzpicture}} \hspace{0.5cm} \text{ and } \hspace{0.5cm}
\left\lfloor \lambda,\S _{1}\right\rfloor =%
\scalebox{0.7}{
 \begin{tikzpicture}
\draw (0,0) rectangle (0.5,0.5);
\draw (0.5,0) rectangle (1,0.5);
\draw (1,0) rectangle (1.5,0.5);
\draw (1.5,0) rectangle (2,0.5);
\draw (2,0) rectangle (2.5,0.5);
\draw (0,0.5) rectangle (0.5,1);
\draw (0.5,0.5) rectangle (1,1);
\draw (0,1) rectangle (0.5,1.5);
\draw (0.5,1) rectangle (1,1.5);
\draw (0,1.5) rectangle (0.5,2);
\node at (0.25,0.25) {0};
\node at (0.75,0.25) {0};
\node at (1.25,0.25) {0};
\node at (1.75,0.25) {0};
\node at (0.25,0.75) {1};
\node at (2.25,0.25) {0};
\node at (0.75,0.75) {1};
\node at (0.25,1.25) {2};
\node at (0.75,1.25) {2};
\node at (0.25,1.75) {3};
\end{tikzpicture}}.
\end{eqnarray*}

\subsubsection{From singular nonsymmetric Jack polynomials to highest weight symmetric Jack polynomials}\label{SubSubSectHWSJP}

The fact that $\left\{ J_{\alpha,\S ^{\prime}}\right\}$ is a basis for $\mathcal{P}_{\tau}$ and that there is a unique symmetric Jack polynomial of degree $\sum_i\left(i-1\right)\tau_{i}$ implies that any symmetric polynomial of that degree is a scalar multiple of it. This minimal symmetric polynomial is constructed as follows. Consider the set of minimal polynomials of isotype $\tau$, $\left\{h_{\S }:\S \in\tab_\tau\right\}$, described in \eqref{Eqp_S0}. Then, by the remark above, the polynomial $\displaystyle{F_{\tau}:=\sum_{\S \in\tab_\tau}\frac{1}{\left\langle \S ,\S \right\rangle _{0}}h_{\S }\otimes \S }$ is the symmetric Jack polynomial of minimal degree. Since $\displaystyle{\sum_{i=1}^{N}F_{\tau}\mathbf D_{i}}$ is symmetric of lower degree, the sum vanishes. By \cite[Thm. 5.22]{DL2011}, $\displaystyle{F_{\tau}=\sum_{\alpha^{+}=\lambda}c_{\alpha}J_{\alpha,\S _{1}}}$, with $\lambda=[(l-1)^{\tau_{l}},(l-2)^{n_{0}+1},\dots,1^{n_{0}+1},0^{(n-1}]$. One of the terms in $F_{\tau}$ is a nonzero multiple of $x^{\lambda}\otimes \S _{1}$ and there is no term of the form $x^{\lambda}\otimes \S ^{\prime}$ with $\S ^{\prime}\neq \S _{1}$ because the other terms in the $J_{\alpha,\S _{1}}$ are of the form $x^{\beta}\otimes \S ^{\prime}$ with $\lambda\vartriangleright\beta$. This implies that one of the terms of $h_{\S _{1}}$ is indeed $x^{\lambda}$. 

Now, we specialize to the hypotheses of Theorem \ref{Jsing} and apply Proposition \ref{SPcomm}. Let $\tau$ be a partition associated with a set of singular polynomials for $\kappa=\kappa_{0}$. By \eqref{defJsing} the leading term of $J_{\alpha\left(\S _{1}\right)}$ is $x^{\mu}$, where $\mu=qs( m_{0},n_{0}-1,m, l-1, n-1, \tau_l )$.
This implies that $\displaystyle{F_{\tau}\rho=\sum_{\S \in\tab_\tau}\frac{1}{\left\langle \S ,\S \right\rangle _{0}}h_{\S }J_{\alpha\left(\S \right)}}$ has leading term $x^{\mu+\lambda}$ (with a nonzero coefficient). Furthermore, $\sum_i F_{\tau}\rho\mathbf D_{i}=0$, for $\kappa =\kappa_{0}$. Thus, $F_{\tau}\rho$ is a highest weight symmetric Jack polynomial \cite{JL}. 
Note that $\mu+\lambda=qs(m_{0}+1,n_{0}-1,m+1,l-1,n-1,\tau_{l})$.

\subsection{The Hecke algebra case}\label{SubHeckealgebra}
We proceed to analyze the Hecke algebra case. To do so, for the rest of the discussion we assume the following result, whose proof is included in a forthcoming paper \cite{CD2019}.
\begin{thm}\label{Conjecture}
 Let $m = dm_0$ and $ n= dn_0 $, for some $d \geq 1$, and $g=\gcd(m_0,n_0)$. Then, the polynomial
$M_{qs ( m_0 ,n_0-1 ,m,l- 1,n -1,\tau_l )}$ is singular for $\displaystyle{(q,t) =(z u^{ -{n_0\over g}} , u^{m_0\over g})}$, where $z^{m_0\over g}$ is a $g^{th}$ root of unity. Equivalently, $\displaystyle{z = e^{\frac{2\pi ik}{m_0}}}$, where $\gcd(k,g)=1$.
\end{thm}

This theorem shows a very natural phenomenon. Hypothesize that there exist singular polynomials with properties analogous to the group case above. For $m$, $n$ and $d$ as in Theorem \ref{Conjecture}, consider 
\begin{eqnarray*}
\tau &=& \left(n-1,\left(n_{0}-1\right)^{l-2},\tau_{l}\right) \vdash N, \text{ with } 1\leq\tau_{l}\leq n_{0}-1, \\
\mu &=&\left(\left(m+\left(l-2\right)m_{0}\right)^{\tau_{l}},\left(m+\left(l-3\right)m_{0}\right)^{n_{0}-1},\ldots,m^{n_{0}-1},0^{n-1}\right).\end{eqnarray*}
Then, the theorem states that $M_{\mu}$ is singular for certain $(q,t)$ satisfying $q^{m}t^{n}=1$. 

%Note that we include the case in which $q^mt^n=1$ and also $q^{m_{1}}t^{n_{1}}=1$ for some $\left( m_{1},n_{1}\right)$ with $\displaystyle{\frac{m_{1}}{n_{1}}=\frac{m}{n}}$. Moreover, the conjectured values of $\left(q,t\right)$ can be written as $q=zu^{-n_{1}},t=u^{m_{1}}$, where $\displaystyle{\frac{m_{1}}{n_{1}}=\frac{m}{n}}$ and $z$ is a root of unity such that $z^{m_{0}}=1$ and $q^{m}t^{-n}=1$. 

As an example, take $N=6$, $m=m_{0}=3$, $n=n_{0}=3$ and $\tau=\left(2,2,2\right)$. Then, consider the quasistaircase $\mu=qs(3,2,3,2,2,2)=\left(6,6,3,3,0,0\right)$. The Macdonald polynomial $M_\mu$ is singular for $q=z u^{-1},t=u$, with $z^{3}=1, z\neq1$. Another example is the quasistaircase $\mu=(24^3,16^3,0^{11})$, for which the Macdonald polynomial $M_\mu$ is singular for the specializations $q^8t^4=1$ and $q^2t=\pm i$.

\subsubsection{Singularity of quasistaircase nonsymmetric Macdonald polynomials}

In general, setting $\displaystyle{q=t^{\frac{1}{\kappa}}}$ and letting $t\rightarrow1$ in $M_{\alpha}$ produces the Jack polynomial $J_{\alpha}$. This should hold for singular polynomials, but not when we specialize to $q=zu^{-n_1}$ and $t=u^{m_1}$, with $z\neq 1$. 

For instance, take $\lambda=(2,0)$. Then, $M_\lambda$ is singular for $qt=-1$, $q^2t^2=1$ but not for $qt=1$. Looking at the first specialization, $qt=-1$, the Macdonald polynomials is $\displaystyle{\left.M_\lambda\right|_{qt=-1} = -(tx_1-x_2)(x_1+x_2)}$ and there is no corresponding Jack polynomial. 
This property indicates that the structure (labeling, isotype, etc.) of singular Macdonald polynomials is fairly close to that of the singular Jack polynomials, but still different enough to be interesting to study. 
\begin{thm}\label{Hsingp}
Consider the singular polynomial $M_{\mu}$ for the specialization $q^{m}t^{n}=1$ and the quasistaircase $\mu$ described in Theorem \ref{Conjecture} (and possibly also a more restrictive condition). Then, $M_{\mu}$ is of isotype $\tau$ and $\left\{M_{\alpha\left(\S \right)}:\S \in\tab_\tau\right\}$, where $\alpha(\S )$ is the reverse lattice permutation associated to $\S $, is a basis for polynomials of isotype $\tau$ for $\mathcal{H}_{N}\left(t\right)$.
\end{thm}

\begin{proof}
We start computing the spectral vector of $\alpha(\S)$:
\begin{eqnarray*}
\zeta_{\alpha(\S)}(i) =q^{\alpha(\S)_{i}}t^{N-r\left(\alpha(\S ),i\right)}=z^{\alpha\left(\S \right)_{i}}u^{\left(N-r\left(\alpha(\S ),i\right)\right)m_{1}-n_{1}\alpha(\S )_i}=u^{m_{1}\ct{S}{i}}=t^{\ct{S}{i}}.
\end{eqnarray*}
It is clear that if $\row{S}{i}=\row{S}{i+1}$, then $\alpha(\S)_{i}=\alpha(\S)_{i+1}$ and $M_{\alpha(\S)}s_{i}=M_{\alpha(\S)}$.
In particular $M_{\mu}T_{i}=tM_{\mu}$, whenever $\mu s_{i}=\mu$. That is, $s_{i}$ is a generator of the stabilizer group of $\mu$, isomorphic to the product $\mathcal{S}_{\tau_{1}}\times\mathcal{S }_{\tau_{2}}\times\cdots \times\mathcal{S }_{\tau_{l}}$. Thus $M_{\mu}$ is in the $\mathcal{H}_{N}(t)$-module induced up from the trivial representation of $\mathcal{H}_{\tau_{1}}(t) \times\mathcal{H}_{\tau_{2}}(t)\times\cdots\times\mathcal{H}_{\tau_{l}}(t)$. By the results of \cite{DJ1986}, this is a direct sum of components labeled by partitions $\gamma$ each of which satisfies $\gamma\succeq\tau$. Furthermore, $\gamma\in\pars{N}$ and $\gamma\succ\tau$, which implies $\sum\limits_{\left(i,j\right)\in\gamma}\left(j-i\right)>\sum\limits_{(i,j)\in\tau}\left(j-i\right)$. By a result of Macdonald \cite[(1.15)]{Macd1995}, it suffices to prove this for $\gamma$, $\gamma^{\prime}$ such that $\gamma_{k}=\gamma_{k}^{\prime}+1$ and $\gamma_{m}=\gamma_{m}^{\prime}-1$, for some $k<m$ (a \textit{raising} operator). Then,
\begin{eqnarray*}
\sum\limits_{(i,j) \in\gamma}\left(j-i\right)-\sum\limits_{\left(i,j\right)\in\gamma^{\prime}}\left(j-i\right) =\left(\gamma_{k}^{\prime}+1-k\right)-\left(\gamma_{m}^{\prime}-m\right)=\left(\gamma_{k}^{\prime}-\gamma_{m}^{\prime}\right)+\left(m-k+1\right)\geq2.
\end{eqnarray*}
The operator $\prod\limits_{i=1}^{N}\phi_{i}$ has the eigenvalue $t^{p}$, for $p=\sum\limits_{(i,j)\in\gamma^{\prime}}\left(j-i\right)=\sum_i\textsf{CT}_{\mathbb{S}_1}$, on polynomials of isotype $\gamma$. Since $M_{\alpha}\prod\limits_{i=1}^{N}\phi_{i}=t^{p}M_{\alpha}$, and using the above inequality for eigenvalues, we see that $M_{\alpha}$ is of isotype $\tau$.

Suppose $\alpha(\S )_{i}<\alpha(\S )_{i+1}$. We know that $\row{S}{i}<\row{S}{i+1}$ and $\S^{(i)}\in\tab_\tau $. Then,
\begin{align*}
\S \tau\left(T_{i}\right) & =\S ^{(i)}-\dfrac{1-t}
{1-t^{\ct{S}{i+1}-\ct{S}{i}}}\S ,\\
M_{\alpha(\S )}T_{i} & =M_{\alpha(\S )s_{i}}-\frac{1-t}{1-t^{\ct{S}{i+1}-\ct{S}{i}}}M_{\alpha},
\end{align*}
 Therefore, $M_{\alpha\left(\S \right)}$ and $\S $ transform the same way under $T_{i}$. 

By a similar argument as in Proposition \ref{Jsing}, suppose $\col{S}{i}=\col{S}{i+1}$. Then, $t^{1+\ct{S}{i}}=t^{\ct{S}{i+1}}$ and
\begin{align*}
T_{i}\phi_{i} & =\frac{1}{t}\left(T_{i}^{2}\phi_{i+1}T_{i}\right)=\frac{1}{t}\left(\left(t-1\right)T_{i}+t\right)\phi_{i+1}T_{i}=\left(t-1\right)\phi_{i}+\phi_{i+1}T_{i},\\
T_{i}\phi_{i+1}& =\left(1-t\right)\phi_{i}+\phi_{i}T_{i}.
\end{align*}
This implies that
\begin{multline*}
\left(M_{\alpha(\S )}+M_{\alpha(\S )}T_{i}\right)\phi_{i}=t^{\ct{S}{i}}M_{\alpha(\S )}+\left(t-1\right)t^{\ct{S}{i}}M_{\alpha(\S )}+t^{\ct{S}{i+1}}M_{\alpha(\S )}T_{i}= \\=t^{\ct{S}{i+1}}\left(M_{\alpha(\S )}+M_{\alpha(\S )}T_{i}\right).
\end{multline*}
Similarly $\left(M_{\alpha(\S )}+M_{\alpha(\S )}T_{i}\right)\phi_{i+1}=t^{\ct{S}{i}}\left(M_{\alpha(\S )}+M_{\alpha(\S )}T_{i}\right)$ and, together with $M_{\alpha(\S )}T_{i}$ being of isotype $\tau$, implies $M_{\alpha(\S )}+M_{\alpha(\S )}T_{i}=0$.
\end{proof}

\begin{cor}
The polynomials obtainable from $M_{\mu}$ by a sequence of maps $f\rightarrow af+bfT_{i}$ are also singular. 
\end{cor}
\begin{proof}
This is a consequence of the fact that $f\mathcal{D}_{i}=0$, for all $i$, implies $fT_{j}\mathcal{D}_{i}=0$, for all $i,j$, which follows from equations (3.8) and (3.9) in \cite{DL2012}.
\end{proof}

\subsubsection{From singular nonsymmetric Macdonald polynomials to highest weight symmetric Macdonald polynomials}
In \cite{DL2012} this case is already studied and we list a few useful results that are proved there. 
\begin{prop}[\cite{DL2012}]
For each $\left(\alpha,\S \right)\in\comps\times$ $\tab_\tau$, there exists a simultaneous eigenfunction $M_{\alpha,\S }$ of the operators $\mathbf E_{i}$ of the form
\begin{eqnarray*}
M_{\alpha,\S }\xi_{i} =\zeta_{\alpha,\S }\left(i\right)M_{\alpha,\S }, \text{ where } \zeta_{\alpha,\S }(i)=q^{a_{i}}t^{\ct{S}{r_(\alpha,i)}} \text{ for } 1\leq i\leq N.
\end{eqnarray*}
Furthermore, the set $\left\{M_{\alpha,\S }\right\}$ is a basis for $\mathcal{P}_{\tau}$.
\end{prop}

The structure of symmetric polynomials in $\mathcal{P}_{\tau}$ is very similar to the group case.
\begin{prop}[\cite{DL2012}]
The polynomial $M_{\alpha,\S }$ can be transformed to $M_{\beta,\S ^{\prime}}$ by a sequence of maps of the form $f\rightarrow af+bfT_{i}$ if and only if $\left\lfloor \alpha,\S \right\rfloor =\left\lfloor \beta,\S ^{\prime}\right\rfloor $.
\end{prop}

\begin{thm}[\cite{DL2012}]
For $\left(\alpha,\S \right)\in\comps\times\tab_\tau$ the $\mathrm{span}\left\{ M_{\beta,\S^\prime}:\left(\beta,\S ^{\prime}\right)\in\mathcal{T}\left(\alpha,\S \right)\right\} $ contains a unique nonzero symmetric polynomial if and only if $\left\lfloor\alpha,\S \right\rfloor$ is a column-strict tableau.
\end{thm}

As in the group case there is a unique symmetric polynomial of minimal degree. In this case, $\sum_i\mathbf{E}_{i}$ determines $\left\lfloor \lambda,\S\right\rfloor $ and the set $\mathcal{T}\left(\lambda,\S \right)$. Indeed $M_{\lambda,\S}\sum_i \mathbf {E}_{i} = \sum_i q^{\lambda_{i}}t^{\ct{S}{i}}M_{\lambda,\S}$. The tableau $\left\lfloor \lambda,\S\right\rfloor $ can be constructed from the eigenvalues. 

As example, consider $\gamma=\left(2,2,2,1,0,0\right)$ and $\S =\scalebox{0.7}{\begin{tikzpicture}
\draw (0,0) rectangle (0.5,0.5);
\draw (0.5,0) rectangle (1,0.5);
\draw (1,0) rectangle (1.5,0.5);
\draw (0,1) rectangle (0.5,1.5);
\draw (0,0.5) rectangle (0.5,1);
\draw (0.5,0.5) rectangle (1,1);
\node at (0.25,0.25) {6};
\node at (0.75,0.25) {5};
\node at (1.25,0.25) {3};
\node at (0.25,1.25) {2};
\node at (0.25,0.75) {4};
\node at (0.75,0.75) {1};
\end{tikzpicture}}$. Then,
$\left\lfloor \gamma,\S \right\rfloor =
\scalebox{0.7}{\begin{tikzpicture}
\draw (0,0) rectangle (0.5,0.5);
\draw (0.5,0) rectangle (1,0.5);
\draw (1,0) rectangle (1.5,0.5);
\draw (0,1) rectangle (0.5,1.5);
\draw (0,0.5) rectangle (0.5,1);
\draw (0.5,0.5) rectangle (1,1);
\node at (0.25,0.25) {0};
\node at (0.75,0.25) {0};
\node at (1.25,0.25) {2};
\node at (0.25,0.75) {1};
\node at (0.25,1.25) {2};
\node at (0.75,0.75) {2};
\end{tikzpicture}}\ $ and 
\begin{eqnarray*}
M_{\gamma,\S} \sum_i \mathbf{E}_{i}=q^{2}\left(1+t^{-2}+t^{2}\right)+qt^{-1}+t+1.
\end{eqnarray*}
Since the entries 1, 2 and 3 in $\S$ can be arbitrarily permuted, there are 6 different RSYTs leading to the same $\left\lfloor \gamma,\S\right\rfloor$. Moreover, the diagram $\left\lfloor \gamma,\S\right\rfloor$ is column-strict, and therefore there is a symmetric polynomial in the $\mathcal{S}_N$-span. Finally, for the shape of $\S$, $(3,2,1)$, the corresponding minimal type is $\lambda=(2,1,1,0,0,0)$, and $\S_1 =\scalebox{0.7}{\begin{tikzpicture}
\draw (0,0) rectangle (0.5,0.5);
\draw (0.5,0) rectangle (1,0.5);
\draw (1,0) rectangle (1.5,0.5);
\draw (0,1) rectangle (0.5,1.5);
\draw (0,0.5) rectangle (0.5,1);
\draw (0.5,0.5) rectangle (1,1);
\node at (0.25,0.25) {6};
\node at (0.75,0.25) {5};
\node at (1.25,0.25) {4};
\node at (0.25,1.25) {1};
\node at (0.25,0.75) {3};
\node at (0.75,0.75) {2};
\end{tikzpicture}}$ and
$\left\lfloor \lambda,\S_1 \right\rfloor =
\scalebox{0.7}{\begin{tikzpicture}
\draw (0,0) rectangle (0.5,0.5);
\draw (0.5,0) rectangle (1,0.5);
\draw (1,0) rectangle (1.5,0.5);
\draw (0,1) rectangle (0.5,1.5);
\draw (0,0.5) rectangle (0.5,1);
\draw (0.5,0.5) rectangle (1,1);
\node at (0.25,0.25) {0};
\node at (0.75,0.25) {0};
\node at (1.25,0.25) {0};
\node at (0.25,0.75) {1};
\node at (0.25,1.25) {2};
\node at (0.75,0.75) {1};
\end{tikzpicture}}\ $.

Let $\left\{ h_{\S }:\S \in\tab_\tau\right\}$ be the set of minimal polynomials of isotype $\tau$ for $\mathcal{H}_{N}\left(1/t\right)$ described in \eqref{Eqp_S0}. By Theorem \ref{Hsymm}, $\displaystyle{F_{\tau}:=\sum_{\S \in\tab_\tau}\frac{1}{\gamma\left(\S ;t\right)}h_{\S }\otimes \S }$ is the symmetric polynomial and eigenfunction of $\sum_i \mathbf E_{i}$ of minimal degree. By \cite[Thm. 5.27]{DL2012}, we can write $\displaystyle{F_{\tau}=\sum_{\alpha^{+}=\lambda} c_{\alpha}M_{\alpha,\S _{1}}}$. One of the terms in $F_{\tau}$ is (a nonzero multiple of) $x^{\lambda}\otimes \S _{1}$ and there is no term of the form
$x^{\lambda}\otimes \S ^{\prime}$ with $\S ^{\prime}\neq \S _{1}$ because the other terms in the $M_{\alpha,\S _{1}}$ are of the form $x^{\beta}\otimes \S ^{\prime}$ with $\lambda\vartriangleright\beta$. This implies that one of the terms of $h_{\S _{1}}$ is indeed $x^{\lambda}$. 

Now, we specialize to the hypotheses of Theorem \ref{Hsingp} and apply Theorem \ref{Hrhocomm}. Let $\tau$ be a partition associated with a set of singular polynomials for $\left( q,t\right)$. By the construction in formula \eqref{defJsing} the leading term of $M_{\alpha\left( \S _{1}\right)}$ is $x^{\mu}$.
This implies that $\displaystyle{F_{\tau}\rho=\sum\limits_{\S \in\tab_\tau}\frac{1}{\gamma\left(\S ;t\right)}h_{\S }M_{\alpha(\S )}}$ has the leading term $x^{\mu+\lambda}$ (with a nonzero coefficient), where the partition $\mu+\lambda$ is the quasistaircase $\mu+\lambda=qs(m_{0}+1,n_{0}-1,m+1,l-1,n-1,\tau_{l})$.

 It is known that $T_{i}$ and $\mathcal D_{j}$ commute for $j\neq i, i+1$. Therefore, $\displaystyle{\sum_{j=1}^{i-1}\mathcal{D}_{j}+\sum_{j=i+2}^{N}\mathcal{D}_{j}}$ commutes with $T_{i}$. Using that $t\displaystyle{\mathcal{D}_{i}=T_{i}\mathcal{D}_{i+1}T_{i}}$, we deduce the other two cases:
 \begin{eqnarray*}
 \mathcal{D}_{i}T_{i}&=&\frac{1}{t}T_{i}\mathcal{D}_{i+1}T_{i}^{2}=\frac{1}{t}T_{i}\mathcal{D}_{i+1}\left\{\left(t-1\right)T_{i}+t\right\}, \\
 \mathcal{D}_{i}T_{i}+\frac{1}{t}T_{i}\mathcal{D}_{i+1}T_{i} &=&T_{i}\mathcal{D}_{i+1}T_{i}+T_{i}\mathcal{D}_{i+1}=t\mathcal{D}_{i}+T_{i}.
 \end{eqnarray*}
 Suppose that $fs_{i}=f$, so $fT_{i}=tf$. Then, apply the operators in the equation to $f$ to get that $f\mathcal{D}_{i}T_{i}+f\mathcal{D}_{i+1}T_{i}=tf\left(\mathcal{D}_{i}+\mathcal{D}_{i+1}\right)$. The commutations imply that if $f$ is symmetric, then $\sum_j f\mathcal{D}_{j}T_{i}=t\sum_j f\mathcal{D}_{j}$. The same computation can be done for the vector-valued operators $\mathbf T_i$ and $\mathbf D_i$, since they satisfy the same relations. This proves the following result. 
\begin{lem}
If $f$ in $\mathcal{P}_{\tau}$ or in $\mathcal{P}$ is symmetric, then so is $\sum_if\mathbf{D}_{i}$ or $\sum_if\mathcal{D}_{i}$, respectively.
\end{lem}
As in the group case, we see that $F_{\tau}\sum_i\mathbf{D}_{i}=0$ because $F_{\tau}$ is the minimal degree symmetric polynomial in $\mathcal{P}_{\tau}$. 
\begin{thm} 
The polynomial $F_{\tau}$ is proportional to the symmetric Macdonald polynomial $\mathrm M_{\mu+\lambda}$ and is annihilated by $\sum_i\mathbf{D}_{i}$. Hence, it is a highest weight Macdonald polynomial.
\end{thm}
 In particular, we have recovered that the symmetric Macdonald polynomial $\mathrm M_{\mu+\lambda}$ is highest weight \cite[Theorem II]{JL} from the singularity of the nonsymmetric Macdonald polynomial $M_\mu$.
This shows that the study of singular nonsymmetric Macdonald polynomials is relevant for the understanding of the Bernevig and Haldane conjectures \cite{BH2008}.

\subsubsection{A note on specializations}\label{SubsubFinalRemark}
From Proposition \ref{singPhi}, a polynomial $p$ is singular if and only if $p\xi_{i}=p\phi_{i}$, for $1\leq i\leq N$.
Now, suppose $M_{\alpha}$ is singular for some specific $\left(q_{0},t_{0}\right)$.  In this case, instead of considering the sum  of $\xi$ as in the symmetric case, we consider the product. Since $\sum_i\left(  N-r(\alpha,i)\right)=N^{2}-\left(  (1+2+\cdots+N\right)  =\frac{1}{2}N\left(  N-1\right)$, we have that
\begin{eqnarray*}
M_{\alpha}\prod\limits_{i=1}^{N}\xi_{i}  =M_{\alpha}\prod\limits_{i=1}^{N}\phi_{i}=\prod\limits_{i=1}^{N}\zeta_{\alpha}\left(  i\right)  M_{\alpha}=q^{\left\vert \alpha\right\vert }t^{N\left(  N-1\right)  /2}M_{\alpha}.
\end{eqnarray*}
Suppose further that $M_{\alpha}$ (specialized at $\left( q_{0},t_{0}\right)$) is of isotype $\tau$, for some partition $\tau$ of $N$. Then, the eigenvalue of $\prod\limits_{i=1}^{N}\phi_{i}$ acting on an $\mathcal{H}_N(t)$-module of isotype $\tau$ is $t^{\Sigma\left(\tau\right)}$, where
\begin{eqnarray*}
\Sigma\left(\tau\right)=\sum\limits_{\left(i,j\right)\in\tau}\left(j-i\right)=\frac{1}{2}\sum_{i=1}^{\ell(\tau)}\tau_{i}\left(\tau_{i}-2i+1\right).
\end{eqnarray*}
Thus $\left(q_{0},t_{0}\right)$ must satisfy the equation $q^{\left\vert \alpha\right\vert }t^{N\left(  N-1\right)  /2-\Sigma\left(\tau\right)  }=1$.

In particular, we look at the quasistaircase partition $qs\left(m,n-1,dm,dn-1,l-1,\tau_{l}\right)$ and the isotype $\tau=\left(dn-1,\left(n-1\right)^{l-2},\tau_{l}\right)$. Then,  
\begin{multline*}
\Sigma\left(\tau\right) =\frac{1}{2}\left(dn-1\right)  \left(dn-2\right)+\frac{1}{2}\sum_{i=2}^{l-2}\left(n-1\right)  \left(n-2i\right)  +\frac{1}{2}\tau_{l}\left(  \tau_{l}-2l+1\right) = \\
 =\frac{1}{2}\left(dn-1\right)\left(dn-2\right)+\frac{1}{2}\left(n-1\right)\left(n-l-1\right)\left(l-2\right) +\frac{1}{2}\tau_{l}\left(\tau_{l}-2l+1\right),
\end{multline*}
and%
\begin{eqnarray*}
\frac{1}{2}N\left(  N-1\right)  -\Sigma\left(  \tau\right)  =\frac{1}{2}n\left(  n-1\right)  \left(  l-2\right)  \left(  l-3+2d\right)  +\tau_{l}n\left(  l-2+d\right). 
\end{eqnarray*}
Denoting $A=\displaystyle{ \frac{1}{n}\left(  \frac{1}{2}N\left(  N-1\right)  -\Sigma\left(  \tau\right)  \right)}$ and after some computations, we get that $(q,t)$ must satisfy $q^{mA}t^{nA}=1$. Note that this is only a necessary condition.

\section{Factorizations at special points}\label{SectionSP}

In this section we examine the polynomials $h_{\S}$ of minimal degree associated to $\S\in\tab_\tau$, defined by formulas \eqref{Eqp_S0} and \eqref{hSgroup} in Section \ref{SubSectPR}, when evaluated at special points. Typically these involve a smaller than $N$ number of free variables. There are also factorizations of highest-weight symmetric Jack and Macdonald polynomials for certain parameter values. The two structures are tied together by the projection formulas of Sections \ref{SubSectProjSingHW} and \ref{SectionQSP}. The key fact is that Macdonald polynomials of isotype $\tau$, denoted by $g_{\S}$ with $\S\in\tab_\tau$ and $\S\neq \S_0$, can be shown to vanish at the special points. This also applies to polynomials $h_{\S}$, with $\S\in \tab_\tau$ of isotype $\tau$ with minimum degree. The proofs are worked out in detail for the Hecke algebra case and a limiting argument ($t\longrightarrow 1$) is used to derive the symmetric group version. % The argument is relatively straightforward for the symmetric group case but is quite intricate in the Hecke algebra case.

\subsection{The symmetric group case}
We start with the definition of a family of relevant polynomials. 
\begin{defi}
For each $\S\in\tab_\tau$, define the \emph{alternating polynomial} by setting 
\begin{eqnarray*}
a_{\S}(x) :=\prod_{j=1}^{\tau_{1}}\prod_{\substack{i,k=1 \\ i<k}}^{\tau_{j}^{\prime}}\left(x_{\S(k,j)}-x_{\S(i,j)}\right),
\end{eqnarray*}
where $\tau^{\prime}$ denotes the transpose of the partition $\tau$ and $\S(i,j)$ denotes the entry in the $i^{\text{th}}$ row and the $j^{\text{th}}$ column. 
\end{defi}

We state the relevant results on evaluations at special points and factorization for the symmetric group case. These are simple consequences of analogous facts for the Hecke algebra. Therefore, references to their proofs are given later on this paper, Remarks \ref{RemarkThmProof} and \ref{RemarkCorProof}. We start describing the specialization. 
\begin{defi}
Let $l:=\ell\left(\tau\right)$ and for $1\leq i\leq l-1$, define the sequence given by
\begin{eqnarray}\label{S1rows}
n_{i}:=N-\sum_{j=1}^{l-i}\tau_{j}+1.
\end{eqnarray}
Moreover, for notational convenience, we set $n_{l}:=N+1$. 
For the free variables $z_{1},\ldots,z_{\tau_l}$, $y_{1},\ldots y_{l-1}$, define the \emph{special point} $\overline{x}\in\mathbb{R}^{N}$ by 
\begin{eqnarray*}
\left(\overline{x}\right)_i = \left\{ 
\begin{array}{ccc}
y_j & & \text{ for } n_j\leq i \leq n_{j+1}-1, \\
z_i  & \hspace{0.5cm} & \text{ otherwise } 
\end{array}\right.
\end{eqnarray*}
\end{defi}
Our first remark is that with the sequence $\{n_i\}$, the rows of $\S_1$ are, from bottom to top, $\left[N,\dots,n_{l-1}\right]$, $\left[n_{l-1}-1,\dots,n_{l-2}\right]$, $\dots$, $\left[n_{2}-1,\dots,n_{1}\right]$, $\left[n_{1}-1,\dots,1\right]$. Our second remark is that we visualize the special point in the following way. For a shape $\tau$, fill the top row with the variables $z_i$ from right to left. Then, for $1\leq j \leq \ell(\tau)-1$, all the entries of the $j^{\text{th}}$ row are filled with $y_{j}$. 
In this way, the special point $\overline{x}$ corresponds to the reading of the tableau from top to bottom and from left to right. 
For instance, for $\tau=(3,3,2)$, 
\begin{eqnarray*}
\scalebox{0.7}{
\begin{tikzpicture}
\draw (0,0) rectangle (0.5,0.5);
\draw (0.5,0) rectangle (1,0.5);
\draw (1,0) rectangle (1.5,0.5);
\draw (0,0.5) rectangle (0.5,1);
\draw (0.5,0.5) rectangle (1,1);
\draw (1,0.5) rectangle (1.5,1);
\draw (0,1) rectangle (0.5,1.5);
\draw (0.5,1) rectangle (1,1.5);
\node at (0.25,0.25) {$y_1$};
\node at (0.75,0.25) {$y_1$};
\node at (1.25,0.25) {$y_1$};
\node at (0.25,0.75) {$y_2$};
\node at (0.75,0.75) {$y_2$};
\node at (1.25,0.75) {$y_2$};
\node at (0.25, 1.25) {$z_2$};
\node at (0.75,1.25) {$z_1$};
\end{tikzpicture}}
\end{eqnarray*}
and the special point is $\overline{x}=\left(z_1,z_2,y_2,y_2,y_2,y_1,y_1,y_1\right)$. 

Now, we present the two results that describe the polynomials $h_{\S}$ at the special point $\left(\overline{x}\right)$, for all $\S\in\tab_\tau$. As mentioned before, the proofs are consequences of the Hecke algebra case, which we investigate in Section \ref{SubSect6.2}.
\begin{thm}\label{ThmHS10}
Let $\S\in\tab_\tau$ be such that $\S\neq \S_1$. Then, $h_{\S}\left(\overline{x}\right)=0$.
\end{thm}

\begin{cor}\label{CorHS1A}
$h_{\S_{1}}\left(\overline{x}\right) =a_{\S_{1}}\left( \overline{x}\right) $.
\end{cor}

Recall the projection of the symmetric Jack polynomial of minimal degree described in Section \ref{SubSubSectHWSJP}: 
\begin{eqnarray*}
F_{\tau}\rho=\sum_{\S\in\tab_\tau}\frac{1}{\left\langle \S,\S\right\rangle _{0}}h_{\S}J_{\alpha\left(\S\right) },
\end{eqnarray*}
where the leading term of $J_{\alpha\left(\S_{1}\right) }$ is $x^{\mu}$ with
\begin{eqnarray*}
\mu=\left( \left(m+\left(l-2\right)m_{0}\right)^{\tau_{l}},\left(m+\left(l-3\right)m_{0}\right)^{n_{0}-1},\ldots,m^{n_{0}-1},0^{n-1}\right).
\end{eqnarray*}
This polynomial is singular for $\kappa=-\frac{m}{n}\notin\mathbb{N}$, with $2\leq n\leq N$, $m_{0}=\frac{m}{\gcd(m,n)}\geq 2$ and $n_{0}=\frac{n}{\gcd(m,n)}\geq 2$. Moreover, it is of isotype $\tau=\left(n-1,\left(n_{0}-1\right)^{l-2},\tau_{l}\right)$, so that $1\leq\tau_{l}=N-\left(n-1\right)-\left(l-2\right)\left(n_{0}-1\right) $.

When we evaluate this polynomial at the special point $\overline{x}$, all the terms but the one with $\S=\S_{1}$ vanish and we obtain
\begin{eqnarray*}
F_{\tau}\rho\left(\overline{x}\right)=\frac{1}{\left\langle \S_{1},\S_{1}\right\rangle }a_{\S_{1}}\left( \overline{x}\right) J_{\mu}\left(\overline{x}\right) .
\end{eqnarray*}
Recall the leading term of $F_{\tau}\rho$ is $x^{\mu+\lambda}$ with $\mu+\lambda=qs(m_{0}+1,n_{0}-1,m+1,l-1,n-1,\tau_{l})$.
Thus, there is a direct relationship between factorizations of $F_{\tau}\rho\left( \overline{x}\right)$, which are highest weight Jack polynomials, and $J_{\mu}\left( \overline{x}\right) $, which is a singular nonsymmetric Jack polynomial.

We finish this case with an example. Let $N=8$, $\kappa=-2/3$ and $\mu=\left(6,6,4,4,2,2,0,0\right)$. Then, $\tau=\left(2^{4}\right)$, $\lambda=\left( 3,3,2,2,1,1,0,0\right)$ and  $\overline{x}=\left(z_1,z_2,y_3,y_3,y_2,y_2,y_1,y_1\right)$, for which
\begin{eqnarray*}
h_{\S_{1}}\left(\overline{x}\right) =  \left( y_{1}-y_{2}\right)^2\left(y_{1}-y_{3}\right)^2\left(y_{2}-y_{3}\right)^{2} \cdot \prod_{i=1}^{2}\prod_{j=1}^{3}\left(z_{i}-y_{j}\right).
\end{eqnarray*}
Moreover, $F_{\tau}\rho$ is a specialization of the symmetric Jack polynomial $J_{\lambda+\mu}=J_{(9,9,6,6,3,3,0,0)}$ to $\kappa=-2/3$. 

\subsection{The Hecke algebra case}\label{SubSect6.2}

There is a similarity with the symmetric group case. However, the difficulty arises because in this case there are factors of the form $t^{u}x_{i}-x_{j}$, where the exponent $u$ is related to a fixed tableau. 

Our starting point is the same as for the symmetric group case: the polynomials $\left\{h_{\S}\right\} $ of minimum degree for a given isotype $\tau$  that were defined in \eqref{Eqp_S0} and \eqref{hSgroup}, together with the basis of isotype $\tau$, $\{g_\S: \S\in\tab_\tau\}$ defined in Definition \ref{BasisG}.

The property $g_{\S}T_{i}=-g_{\S}$, for $\S$ such that $\col{S}{i} = \col{S}{i+1}$ has an important consequence that we describe in the following result. 
\begin{prop}\label{gdiv}
Let $\S\in\tab_\tau$ be such that $\col{S}{i} = \col{S}{i+1}$. Then, $g_{\S}(x)$ is divisible by $(tx_{i}-x_{i+1})$ and the quotient $\displaystyle{\frac{g_{\S}(x)}{tx_{i}-x_{i+1}}}$ is symmetric in $x_i$ and $x_{i+1}$. 
Furthermore for $i<j\leq m$ such that $\col{S}{j}=\col{S}{i}$,  $g_{\S}(x)$ is divisible by $\displaystyle{\prod_{i\leq j<k\leq m}\left( tx_{j}-x_{k}\right)}$.
\end{prop}

\begin{proof}
The equation $g_{\S}T_{i}=-g_{\S}$ is equivalent to $\dfrac{g_{\S}(x)}{tx_{i}-x_{i+1}}=\dfrac{g_{S}\left(xs_{i}\right)}{tx_{i+1}-x_{i}}$;, which is symmetric in $x_i$ and $x_{i+1}$. Moreover, $\dfrac{g_{\S}(x)}{tx_{i}-x_{i+1}}\left(tx_{i+1}-x_{i}\right) =g_{\S}\left(xs_{i}\right) $ is a polynomial, and thus by the unique factorization property $g_{\S}(x)$ is divisible by $tx_{i}-x_{i+1}$. 

If $g_{\S}(x)$ is divisible by $\left(tx_{j}-x_{k}\right)$ and $k+1\leq m$, then the symmetry of $\dfrac{g_{\S}(x)}{tx_{k}-x_{k+1}}$ under $s_{k}$ shows that $g_{\S}(x)$ is divisible by $\left(tx_{j}-x_{k+1}\right)$. A similar argument applies to the situation in which $i\leq j<k-1\leq m-1$, and therefore, it shows that $g_{\S}(x)$ is divisible by $\left(tx_{j+1}-x_{k}\right) $. Argue inductively first for $k=i+1,\ldots,m$ to demonstrate the factors $\left(tx_{i}-x_{k}\right) $ and then show that $\left(tx_{i+1}-x_{k}\right) ,\left( tx_{i+2}-x_{k}\right) ,\ldots$ are factors.
\end{proof}

Let us describe a modified version of the alternating polynomials defined for the symmetric group case. 
\begin{defi}\label{fSdef}
Let $\S\in\tab_\tau$ and $(i,j)$ be a pair such that $i<j$ and $\col{S}{i}=\col{S}{j}$. We denote by $R_\S(i,j)$ the number of $k>i$ such that $\row{S}{k}=\row{S}{j}$ and $\col{S}{k}\geq \col{S}{i}$. The \emph{modified alternating polynomials} are defined as
\begin{eqnarray*}
f_{\S}(x) :=\prod_{\substack{1\leq i<j\leq N\\ \col{S}{i}=\col{S}{j}}}\left(t^{R_\S(i,j)}x_{i}-x_{j}\right).
\end{eqnarray*}
\end{defi}
For example, consider $\S=\scalebox{0.7}{
\begin{tikzpicture}
\draw (0,0) rectangle (0.5,0.5);
\draw (0.5,0) rectangle (1,0.5);
\draw (1,0) rectangle (1.5,0.5);
\draw (1.5,0) rectangle (2,0.5);
\draw (2,0) rectangle (2.5,0.5);
\draw (0,0.5) rectangle (0.5,1);
\node at (0.25,0.25) {6};
\node at (0.75,0.25) {5};
\node at (1.25,0.25) {4};
\node at (1.75,0.25) {2};
\node at (2.25,0.25) {1};
\node at (0.25,0.75) {3};
\end{tikzpicture}}$. Then, $R_{\S}(3,6)=3$ and then $f_{\S}(x)=t^{3}x_{3}-x_{6}$. Observe that for $\S_{0}$ we have that $R_{\S_0}(i,j)=1$, for all $i<j$ such that $\col{S_0}{i}=\col{S_0}{j}$, and then $f_{\S_0}(x)=h_{\S_0}(x)$.

Note that $f_{\S}(x)=a_{\S}(x)$ in the limiting symmetric group case $t\longrightarrow 1$.
The following definition is very useful for our study.
\begin{defi}
For $\S\in\tab_\tau$, we say that $\left(\S,\S^{(i)}\right)$ is an \emph{adjacent pair} if $\row{S}{i}<\row{S}{i+1}$.
\end{defi}
The following result shows how $f_{\S^{(i)}}$ can be computed by using $f_\S$ for adjacent pairs since $R_\S(i+1,u)=b$ and $R_{\S^{(i)}}(i,u) = b+1$, with $b=\col{S}{i}-\col{S}{i+1}$ and $u=\S(\row{S}{i},\col{S}{i+1})$. 
\begin{lem}\label{fSS}
Let $\S\in\tab_\tau$ be such that $\left(\S,\S^{(i)}\right) $ is an adjacent pair. Let $b=\col{S}{i}-\col{S}{i+1}$ and $u$ be entry in the position $(\row{S}{i},\col{S}{i+1})$, that is $u=\S(\row{S}{i},\col{S}{i+1})$. Then, 
\begin{eqnarray*}
f_{\S^{(i)}}(x)=f_{\S}(x)s_i\cdot \frac{t^{b+1}x_{i}-x_{u}}{t^{b}x_{i}-x_{u}}.
\end{eqnarray*}
\end{lem}

Now we are ready to introduce the specialization, which in this case depends on $\Y\in\rst_\tau$.
\begin{defi}
For the free variables $y_{1}$, $y_{2}$, $\dots$, $y_{l-1}$, and $z_{1}$, $\dots$, $z_{\tau_{l}}$, we define the \emph{special point at $\Y\in\rst_\tau$} by
\begin{eqnarray*}
\overline{x}(\Y)_k = \left\{
\begin{array}{ll}
t^{1-\col{Y}{k}}y_{\row{Y}{k}}, & \text{ for } 1\leq \row{Y}{k}< l, \\
z_{\tau_l+1-\col{Y}{k}}, & \text{ for } \row{Y}{k}=l.
\end{array}
\right.
\end{eqnarray*}
\end{defi}
For example, consider $\Y=%
\scalebox{0.7}{
\begin{tikzpicture}
\draw (0,0) rectangle (0.5,0.5);
\draw (0.5,0) rectangle (1,0.5);
\draw (1,0) rectangle (1.5,0.5);
\draw (0,0.5) rectangle (0.5,1);
\draw (0.5,0.5) rectangle (1,1);
\node at (0.25,0.25) {5};
\node at (0.75,0.25) {3};
\node at (1.25,0.25) {2};
\node at (0.25,0.75) {4};
\node at (0.75,0.75) {1};
\end{tikzpicture}}$ .Then, $\overline{x}\left(\Y\right) =\left( z_{1},t^{-2}y_{1},t^{-1}y_{1},z_{2},y_{1}\right)$. 

We point out that the various polynomials $f_\S$, $g_\S$, and $h_\S$ are indexed by $\S\in\tab_\tau$, while the special points are labeled by $\Y\in\rst_\tau$. Recall also that $\tab_\tau\subset \rst_\tau$. Moreover, if $\Y\in\rst_\tau$ and $\row{Y}{i}\neq \row{Y}{i+1}$, then $\Y^{(i)}\in\rst_\tau$ and $\overline{x}\left(\Y^{(i)}\right) =\overline{x}(\Y)s_{i}$. 

Here is graphical representation of the idea. We built $\mathrm{T}_\tau$, for $\tau\in\pars{N}$ by filling the boxes of $\tau$ in the following way: $\mathrm{T}_\tau (i,j)=t^{1-j}y_{i}$, for $1\leq i<l$ and $1\leq j\leq\tau_{i}$, and $\mathrm{T}_\tau(l,j)=z_{\tau_{l}+1-j}$, for $1\leq j\leq\tau_{l}$. For instance, for $\tau=(4,2,2)$,
\begin{center}
\vspace{-0.7cm}\begin{tikzpicture}
\node at (-0.7,1) {$\mathrm{T}_\tau =$}; 
\scalebox{0.7}{
\draw (0,0) rectangle (0.7,1);
\draw (0.7,0) rectangle (2,1);
\draw (2,0) rectangle (3.4,1);
\draw (3.4,0) rectangle (4.8,1);
\draw (0,1) rectangle (0.7,2);
\draw (0.7,1) rectangle (2,2);
\draw (0,3) rectangle (0.7,2);
\draw (0.7,3) rectangle (2,2);
\node at (0.35,0.5) {$y_{1}$};
\node at (1.35,0.5) {$t^{-1}y_{1}$};
\node at (2.7,0.5) {$t^{-2}y_{1}$};
\node at (4.1,0.5) {$t^{-3}y_{1}$};
\node at (0.35,1.5) {$y_{2}$};
\node at (1.35,1.5) {$t^{-1}y_{2}$};
\node at (0.35,2.5) {$z_{2}$};
\node at (1.35,2.5) {$z_1$};}
\end{tikzpicture}
\end{center}
Then, $\overline{x}(\Y)_k = \mathrm{T}_\tau \left(\row{S}{k},\col{S}{k}\right)$, for $1\leq k\leq N$. Notice that $\mathrm{T}_\tau$ is the same for all $\Y\in\rst_\tau$ and that the special point is a permutation of the list of entries of $\mathrm{T}_\tau$. 

We demonstrate vanishing properties of type $h_{\S}\left(\overline{x}(\Y)\right)=0$, for  $\Y\in\rst_\tau$, by using inductive techniques starting with $h_{\S_{0}}$. The basic step relies on the transformation for the adjacent pair $\left(\S,\S^{(i)}\right)$:
\begin{multline}
h_{\S^{(i)}}(x) =h_{\S}T_{i}(x)-\dfrac{t-1}{1-t^{\col{S}{i+1}-\col{S}{i}}}h_{\S}(x)= \\ 
=\left[ \frac{x_{i}\left(1-t\right)}{x_{i}-x_{i+1}}-\dfrac{t-1}{1-t^{\col{S}{i+1}-\col{S}{i}}}\right] h_{\S}(x)+\frac{tx_{i}-x_{i+1}}{x_{i}-x_{i+1}}h_{\S}\left(xs_{i}\right). \label{hxs}
\end{multline}
For the purpose of evaluating this formula at $\overline{x}(\Y)$, note that $\overline{x}(\Y)_i\neq\overline{x}(\Y)_{i+1}$.
\begin{prop}\label{samerow}
Consider $\S\in\tab_\tau$ such that $\left(\S,\S^{(i)}\right)$ is an adjacent pair, and $\Y\in\rst_\tau$. If $h_{\S}\left(\overline{x}(\Y) \right) =0$ and $\row{Y}{i}=\row{Y}{i+1}$, then $h_{\S^{(i)}}\left( \overline{x}(\Y)\right) =0$.
\end{prop}

\begin{proof}
If $\row{Y}{i}<l$, then $\overline{x}(\Y)_{i+1}=t\overline{x}(\Y)_{i}$  by definition, and both terms in \eqref{hxs} vanish. 
If $\row{Y}{i}=l$, then $h_{\S}\left(\overline{x}(\Y)\right) =0$ is an identity in the free
variables $z_{1},z_{2},\ldots,z_{\tau_{l}}$. Moreover, this identity is thus invariant under any permutation of these variables, and so $h_{\S}\left( \overline{x}(\Y)s_{i}\right) =0$ in \eqref{hxs}.
\end{proof}

For $\rst_\tau$, the induction is based on an inversion count. For $\Y\in\rst_\tau$, define an inversion statistics as
$\inv_{0}(\Y):=\#\{(i,j)|\ i<j \text{ and } \row{Y}{i}< \row{Y}{j}\}$. For instance, take $\tau=(3,2)$. For $\Y=\scalebox{0.7}{
\begin{tikzpicture}
\draw (0,0) rectangle (0.5,0.5);
\draw (0.5,0) rectangle (1,0.5);
\draw (1,0) rectangle (1.5,0.5);
\draw (0,0.5) rectangle (0.5,1);
\draw (0.5,0.5) rectangle (1,1);
\node at (0.25,0.25) {5};
\node at (0.75,0.25) {2};
\node at (1.25,0.25) {1};
\node at (0.25,0.75) {4};
\node at (0.75,0.75) {3};
\end{tikzpicture}}$ we have $\inv_0(\Y)=4$ coming from the pairs $(1,3)$, $(1,4)$, $(2,3)$, and $(2,4)$. For $\Y_1=\scalebox{0.7}{
\begin{tikzpicture}
\draw (0,0) rectangle (0.5,0.5);
\draw (0.5,0) rectangle (1,0.5);
\draw (1,0) rectangle (1.5,0.5);
\draw (0,0.5) rectangle (0.5,1);
\draw (0.5,0.5) rectangle (1,1);
\node at (0.25,0.25) {5};
\node at (0.75,0.25) {4};
\node at (1.25,0.25) {3};
\node at (0.25,0.75) {2};
\node at (0.75,0.75) {1};
\end{tikzpicture}}$ we have $\inv_0(\Y_1)=0$. 

Let us see some direct properties of $\inv_0(\Y)$. 
\begin{lem} 
Let $\Y\in \rst_\tau$. 
\begin{itemize}
\item If $\row{Y}{i}\neq \row{Y}{i+1}$, then $\inv_0(\Y^{(i)}) = \inv_0(\Y)\pm 1$.
\item For $\Y=\S_{1}$, $\inv_{0}\left(\Y\right) =0$. 
\item If $\left(\Y,\Y^{(i)}\right)$ is an adjacent pair, then $\inv_{0}\left(\Y^{(i)}\right) =\inv_{0}\left(\Y\right) -1$. %since $\row{S}{i} < \row{S}{i+1}$. $\pi(\Y)$ the $i^{\text{th}}$ entry comes before the $(i+1)^{\text{th}}$ entry.
\end{itemize}
\end{lem}

We prove the inductive step first, which is relatively easy.
\begin{prop}\label{vstep}
Consider $\S\in\tab_\tau$ such that $\left(\S,\S^{(i)}\right)$ is an
adjacent pair. If $h_{\S}\left(\overline{x}(\Y)\right) =0$ for each $\Y\in\rst_\tau$ such
that $\Y\neq \S$ and $\inv_{0}\left(\Y\right)\leq\inv_{0}(\S)$, then $h_{\S^{(i)}}\left(\overline{x}\left(\Y^\prime\right)\right) =0$ for each $\Y^\prime \in\rst_\tau$ such that $\Y^\prime \neq \S^{(i)}$ and $\inv_{0}\left(\Y^\prime\right) \leq\inv_{0}\left(\S^{(i)}\right)$.
\end{prop}

\begin{proof}
Consider $\Y^\prime\in\rst_\tau$ satisfying the hypothesis of the statement. Then, $\Y^{\prime (i)}\neq \S$ (or else $\Y^\prime = \S^{(i)}$) and $\Y^\prime \neq \S$ because $\inv_0(\Y^\prime)\leq \inv_0(\S)-1$. Hence, $h_\S(\overline{x}(\Y^\prime))=0$. 

Now, if $\row{Y^\prime}{i}=\row{Y^\prime}{i+1}$, then $h_{\S^{(i)}}\left(\overline{x}(\Y^\prime) \right) =0$ by Proposition \ref{samerow}. 

Otherwise $\Y^\prime s_i\in \rst_\tau$ and $\inv_0(\Y^\prime s_i) = \inv_0(\Y^\prime) \pm 1\leq \inv_0(\S^{(i)}) +1 = \inv_0(\S)$. Thus, $h_{\S}\left(\overline{x}\left(\Y^\prime\right)s_{i}\right) =h_{S}\left(\overline{x}\left(\Y^{\prime(i)}\right) \right) =0$ and $h_{\S^{(i)}}\left( \overline{x}\left(\Y^\prime \right) \right) =0$ by \eqref{hxs}.
\end{proof}

It remains to show the basis of the induction, that is $\S_{0}$ satisfies the hypotheses of the Proposition \ref{vstep}.
\begin{thm}\label{hS0start}
Consider $\Y\in\rst_\tau$ such that $\Y\neq \S_{0}$ and $\inv_{0}\left(\Y\right)\leq\inv_{0}\left(\S_{0}\right)$. Then $h_{\S_{0}}\left( \overline{x}\left(\Y\right) \right)=0$.
\end{thm}

\begin{proof}
For $1\leq n \leq \tau_l$, let $I_{n}$ denote the set of the entries in the $n^{\text{th}}$ column of $\S_0$. Then, we write $\Y$ as the disjoint union $\Y=\displaystyle{\bigcup_j\left(\Y\cap I_j\right)}$, where each $\Y\cap I_j$ is row-ordered skew shape. That is, for $a,b\in \Y\cap I_j$, $a<b$ and $\row{Y}{a} = \row{Y}{b}$ imply that $\col{Y}{a}>\col{Y}{b}$. Let $C(i,j) :=\left\{k\ |\ k\in \Y\cap I_{j} \text{ and } \row{Y}{k}=i\right\} $. Since each entry in $I_n$ is larger than each entry in $I_m$, for $n<m$, and by the row-strictness of $\Y$, each nonempty set $C(i,j)$ consists of contiguous entries. Moreover, when we fix $i$, the sets $C(i,1)$, $C(i,2)$, $\dots$ are not interlaced. 

If for some $i<l$, there exists a $j$ such that $C(i,j)$ has at least two elements, then $h_{\S_0}(\overline{x}(\Y))=0$. That is because there are two adjacent entries $k_1$ and $k_2$, with $k_1>k_2$, that are in the $i^{\text{th}}$ row of $\Y$ and in the $j^{\text{th}}$ column of $\S_0$. Therefore, $h_{\S_0}$ has $(tx_{k_2}-x_{k_1})$ as a factor and it vanishes at $\overline{x}(\Y)$ because $\overline{x}(\Y)_{k_1}= t\cdot  \overline{x}(\Y)_{k_2}$. 

Now suppose $h_{\S_{0}}\left(\overline{x}(\Y)\right)\neq0$. Then, by the previous argument, $C(i,j)$ has at most one element for each $i<l$. Moreover, for $1\leq k \leq \tau_l$, $\Y\cap I_k$ has size $l$ since these are exactly the columns that extend to the bottom of the tableau. Then, $C(l,k)$ cannot be empty (otherwise $C(i,k)$ would have more than one element for some $i<l$), and $C(l,k)$ cannot have more than one element (otherwise $C(l,k^\prime)$ would be empty for some other $k^\prime\leq \tau_l$). Therefore, for $1\leq i \leq l$, $C(i,k)$ has exactly one element. 

Our next step is to prove by induction on $j$ that the $j^{\text{th}}$ column of $\Y$ is a permutation of $I_j$. Suppose that the first $k-1$ columns of $\Y$ are permutations of $I_1$, $\dots$, $I_{k-1}$. Then, $\Y\cap I_k \subset \left\{ [i,j]\ |\  j \geq k \text{ and } 1\leq i \leq \tau^\prime_j\leq \tau_k^\prime\right\}$. Since $C(i,k)$ has at most one element and $\displaystyle{\sum_{i=1}^{\tau_{k}^{\prime}}\#C\left(i,k\right) =\#I_{k}=\tau_{k}^{\prime}}$, it follows that $C\left( i,k\right)$ has exactly one element, also for $1\leq i\leq\tau_{k}^{\prime}$. In fact, $C(i,k)=\Y(i,k)$ or else $C(i,k)=\Y(i,m)$ for some $m>k$. Then, $\Y(i,k)<\Y(i,m)$, violating the row-strictness of $\Y$. 

As a consequence $\S_{0}$ is obtained from $\Y$ by arranging each column in descending order, and thus $\inv_{0}(\Y) >\inv_{0}\left(\S_{0}\right)$, unless $\Y=\S_{0}$.%, because interchanging $i$ and $i+1$ in $\pi(\Y)$ when $i$ is in a lower number row of $\Y$ decreases $\inv_{0}(\Y)$ by 1.

Thus, $h_{\S_{0}}\left( \overline{x}(\Y)\right) \neq0$ implies $\Y=\S_{0}$ or $\inv_{0}\left(\Y\right) >\inv_{0}\left(\S_{0}\right) $.
\end{proof}

For example, for $\tau=(3,3)$, $\S_0=\scalebox{0.7}{
\begin{tikzpicture}
\draw (0,0) rectangle (0.5,0.5);
\draw (0.5,0) rectangle (1,0.5);
\draw (1,0) rectangle (1.5,0.5);
\draw (0,0.5) rectangle (0.5,1);
\draw (0.5,0.5) rectangle (1,1);
\draw (1,0.5) rectangle (1.5,1);
\node at (0.25,0.25) {6};
\node at (0.75,0.25) {4};
\node at (1.25,0.25) {2};
\node at (0.25,0.75) {5};
\node at (0.75,0.75) {3};
\node at (1.25, 0.75) {1};
\end{tikzpicture}}$ and $\inv_0(\S_0) = 3$. Moreover, $h_{\S_0}(x)=(tx_1-x_2)(tx_3-x_4)(tx_5-x_6)$. Now, consider $\Y=\scalebox{0.7}{
\begin{tikzpicture}
\draw (0,0) rectangle (0.5,0.5);
\draw (0.5,0) rectangle (1,0.5);
\draw (1,0) rectangle (1.5,0.5);
\draw (0,0.5) rectangle (0.5,1);
\draw (0.5,0.5) rectangle (1,1);
\draw (1,0.5) rectangle (1.5,1);
\node at (0.25,0.25) {6};
\node at (0.75,0.25) {3};
\node at (1.25,0.25) {1};
\node at (0.25,0.75) {5};
\node at (0.75,0.75) {4};
\node at (1.25, 0.75) {2};
\end{tikzpicture}}$, with $\inv_0(\Y)=5$, and associated special point $\overline{x}(\Y)= \left( y_1t^{-2}, z_1, y_1t^{-1}, z_2,z_1,y_1\right)$. When we evaluate $h_{\S_0}$ at this special point, we get that $h_{\S_0}\left( \overline{x}(\Y)\right) \neq0$. However, if we consider $\Y^\prime=\scalebox{0.7}{
\begin{tikzpicture}
\draw (0,0) rectangle (0.5,0.5);
\draw (0.5,0) rectangle (1,0.5);
\draw (1,0) rectangle (1.5,0.5);
\draw (0,0.5) rectangle (0.5,1);
\draw (0.5,0.5) rectangle (1,1);
\draw (1,0.5) rectangle (1.5,1);
\node at (0.25,0.25) {6};
\node at (0.75,0.25) {4};
\node at (1.25,0.25) {3};
\node at (0.25,0.75) {5};
\node at (0.75,0.75) {2};
\node at (1.25, 0.75) {1};
\end{tikzpicture}}$, for which $\inv_0(\Y^\prime)= 2$, we get that $\overline{x}(\Y^\prime)= \left(z_1,z_2,y_1t^{-2},y_1t^{-1},z_3,y_1\right)$ and that $h_{\S_0}\left( \overline{x}(\Y)\right)=0$. 

\begin{cor}\label{hSSzero}
Let $\S\in\tab_\tau$ and $\Y\in\rst_\tau$ be such that $\Y\neq \S$ and $\inv_{0}\left(\Y\right) \leq\inv_{0}\left(\S\right) $. Then, $h_{\S}\left(\overline{x}\left(\Y\right)\right) =0$.
\end{cor}

\begin{proof}
Since every $\S\in\tab_\tau$ can be reached from $\S_0$ by a sequence of adjacent pair constructions with a total number of $\inv_{0}\left(\S_{0}\right) -\inv_{0}\left(\S\right)$ steps, the conclusion follows from Theorem \ref{hS0start} and Proposition \ref{vstep}.
\end{proof}

\begin{remark}\label{RemarkThmProof}
If we take the limit $t\longrightarrow 1$ and set $\S=\S_1$, we prove the corresponding result for the symmetric case, Theorem \ref{ThmHS10}.
\end{remark}

Since $\tab_\tau\subset \rst_\tau$, we can consider special points for $\S\in\tab_\tau$, for which the specializations are also interesting.  
\begin{thm}\label{heqfS}
For $\S\in\tab_\tau$ then $h_{\S}\left(\overline{x}\left(\S\right) \right) =f_{S}\left( \overline{x}\left(
\S\right) \right) $.
\end{thm}

\begin{proof}
We proceed by induction on the inversion number. The statement is trivially
true for $\S=\S_{0}$. Since every $\S\in\tab_\tau$ can be reached from $\S_0$ by a sequence of adjacent pair constructions, assume that $h_{\S}\left(\overline{x}(\S)\right) =f_{\S}\left(\overline{x}(\S)\right)$ for some $\S$ such that $\left(\S,\S^{(i)}\right) $ is an adjacent pair.

Since $\inv_0(\S^{(i)})= \inv_0(\S)-1$, by Theorem \ref{hSSzero}, $h_{\S}\left(\overline{x}(\S^{(i)})\right) =0$. Now, we evaluate \eqref{hxs} at $\overline{x}\left(\S^{(i)}\right)=\overline{x}\left(\S\right)s_{i}$.
\begin{eqnarray*}
h_{\S^{(i)}}\left(\overline{x}\left(\S^{(i)}\right)\right) =\frac{t\overline{x}\left(\S^{(i)}\right)_{i}-\overline{x}(\S^{(i)})_{i+1}}{\overline{x}(\S^{(i)})_{i}-\overline{x}(\S^{(i)})_{i+1}}h_{\S}\left(\overline{x}(\S^{(i)})s_{i}\right) =\frac{t\overline{x}\left(\S\right)_{i+1}-\overline{x}\left(\S\right)_{i}}{\overline{x}\left(\S\right)_{i+1}-\overline{x}\left(\S\right)_{i}}h_{\S}\left(\overline{x}\left(\S\right) \right) .
\end{eqnarray*}
Let $i=\S(r_1,c_1)$ and $i+1=\S(r_2,c_2)$ and $u=\S(r_1,c_2)$. By hypothesis $r_{1}<r_{2}$ and $c_{1}>c_{2}$, and $\overline{x}\left(\S\right)_{i}=t^{1-c_{1}}y_{r_{1}}$ because $r_{1}<r_{2}\leq l$.
By Lemma \ref{fSS},
\begin{eqnarray*}
f_{\S^{(i)}}\left(\overline{x}\left(\S^{(i)}\right)\right) =\frac{t^{c_{1}-c_{2}+1}\overline{x}\left(\S^{(i)}\right)_{i}-\overline{x}\left(\S^{(i)}\right)_{u}}{t^{c_{1}-c_{2}}\overline{x}\left(\S^{(i)}\right)_{i}-\overline{x}\left(\S^{(i)}\right)_{u}}f_{\S}\left(\overline{x}\left(\S^{(i)}\right)s_{i}\right),
\end{eqnarray*}
Since $\overline{x}\left(\S^{(i)}\right)_{u}=t^{1-c_{2}}y_{r_{1}}$ and $\overline{x}\left(\S^{(i)}\right)_{i}=\overline{x}\left(\S\right)_{i+1}$, we find that
\begin{multline*}
\frac{t^{c_{1}-c_{2}+1}\overline{x}\left(\S^{(i)}\right)_{i}-\overline{x}\left(\S^{(i)}\right)_{u}}{t^{c_{1}-c_{2}}\overline{x}\left(\S^{(i)}\right)_{i}-\overline{x}\left(\S^{(i)}\right)_{u}}=
\frac{t^{c_{1}-c_{2}+1}\overline{x}\left(\S\right)_{i+1}-t^{1-c_{2}}y_{r_{1}}}{t^{c_{1}-c_{2}}\overline{x}\left(\S\right)_{i+1}-t^{1-c_{2}}y_{r_{1}}}=
\frac{t\overline{x}\left(\S\right)_{i+1}-t^{1-c_{1}}y_{r_{1}}}{\overline{x}\left(\S\right)_{i+1}-t^{1-c_{1}}y_{r_{1}}}= \\
= \frac{h_{\S^{(i)}}\left( \overline{x}\left(\S^{(i)}\right)\right)}{h_{\S}\left(\overline{x}\left(\S\right)\right) },
\end{multline*}
which completes the proof.
\end{proof}
\begin{remark}\label{RemarkCorProof}
Again, taking the limit $t\longrightarrow 1$, we prove the corresponding result for the symmetric case stated in Corollary \ref{CorHS1A}.
\end{remark}

We are ready to describe $h_{\S_{1}}\left(\overline{x}\left(\S_{1}\right) \right)$, for which $\overline{x}\left(\S_{1}\right)_{i}=z_{i}$, for $1\leq i\leq\tau_{l}$.
\begin{prop}\label{fS1S1}
Set $E\left(\tau\right) :=\displaystyle{\frac{1}{2}\sum_{i=1}^{l}\left(i-1\right) \tau_{i}\left(\tau_{i}-1\right)}$. Then
\begin{eqnarray*}
h_{\S_{1}}\left(\overline{x}\left(\S_{1}\right) \right) = f_{\S_{1}}\left(\overline{x}\left(\S_{1}\right) \right)  =t^{-E\left(\tau\right) }\prod_{k=1}^{\tau_{l}}\prod_{i=1}^{l-1}\left(t^{\tau_{i}}z_{k}-y_{i}\right)\cdot \prod_{i=1}^{l-2}\prod_{j=i+1}^{l-1}\prod_{k=1}^{\tau_{j}}\left( t^{\tau_{i}-k+1}y_{j}-y_{i}\right) .
\end{eqnarray*}
\end{prop}

\begin{proof}
The first equality comes from Theorem \ref{heqfS}. For the second, we look at the terms in $f_{\S_1}(\overline{x}(\S_1))$ coming from the $k^{\text{th}}$ column, for $1\leq k \leq \tau_l$. For $1\leq i < j<l$, the entries from $\overline{x}\left(\S_{1}\right)$ corresponding to $\S_{1}(i,k)$, $\S_{1}(j,k)$ and $\S_{1}(l,k)$ are $t^{1-k}y_{i}$, $t^{1-k}y_{j}$ and $z_{\tau_{l}+1-k}$, respectively. Moreover, in $\S_{1}$ every entry in the $i^{\text{th}}$ row is larger than any entry in higher numbered rows, and therefore the exponent $R_{\S_1}(\S_{1}(i,k),\S_{1}(j,k))=\tau_{i}-k+1$. Putting this together, the corresponding factors in $f_{\S_{1}}\left(\overline{x}\left(\S_{1}\right)\right)$ are
$\left(t^{\tau_{i}-k+1}t^{1-k}y_{j}-t^{1-k}y_{i}\right) $, with $1\leq
k\leq\tau_{j}$, and $\left( t^{\tau_{i}-k+1}z_{\tau_{l}+1-k}-t^{1-k}%
y_{i}\right) $, with $1\leq k\leq\tau_{l}$. Thus
\begin{eqnarray*}
f_{\S_{1}}\left(\overline{x}\left(\S_{1}\right)\right)=\prod_{k=1}^{\tau_{l}}\prod_{i=1}^{l-1}\left( t^{1-k}\left(t^{\tau_{i}}z_{\tau_{l}+1-k}-y_{i}\right) \right) \cdot \prod_{i=1}^{l-2}\prod_{j=i+1}^{l-1}\prod_{k=1}^{\tau_{j}}\left( t^{1-k}\left( t^{\tau_{i}-k+1}y_{j}-y_{i}\right) \right) .
\end{eqnarray*}
Collecting the $t^{1-k}$ factors and using the sum $\displaystyle{\sum_{k=1}^{\tau_{j}}\left( 1-k\right) =-\frac{1}{2}\tau_{j}\left( \tau_{j}-1\right)}$ we obtain the term $t^{-E\left( \tau\right) }$. Also, change the index $k$ to $\tau_{l}+1-k$ in the first term. 
\end{proof}
For example, take $\tau=\left( 4,3,2,2\right) $ and
$\S_{1}=
\scalebox{0.7}{
\begin{tikzpicture}
\draw (0,0) rectangle (0.5,0.5);
\draw (0.5,0) rectangle (1,0.5);
\draw (1,0) rectangle (1.5,0.5);
\draw (1.5,0) rectangle (2,0.5);
\draw (0,0.5) rectangle (0.5,1);
\draw (0.5,0.5) rectangle (1,1);
\draw (1,0.5) rectangle (1.5,1);
\draw (0,1) rectangle (0.5,1.5);
\draw (0.5,1) rectangle (1,1.5);
\draw (0,1.5) rectangle (0.5,2);
\draw (0.5,1.5) rectangle (1,2);
\node at (0.25,0.25) {11};
\node at (0.75,0.25) {10};
\node at (1.25,0.25) {9};
\node at (1.75,0.25) {8};
\node at (0.25,0.75) {7};
\node at (0.75,0.75) {6};
\node at (1.25, 0.75) {5};
\node at (0.25,1.25) {4};
\node at (0.75,1.25) {3};
\node at (0.25,1.75) {2};
\node at (0.75,1.75) {1};
\end{tikzpicture}}$, 
\begin{multline*}
f_{\S_{1}}\left( \overline{x}\left(\S_{1}\right) \right) =t^{-12}\left(
t^{4}z_{1}-y_{1}\right) \left( t^{2}z_{1}-y_{2}\right) \left( t^{3}%
z_{1}-y_{3}\right) \left( t^{4}z_{2}-y_{1}\right) \left( t^{3}z_{2}%
-y_{2}\right) \left( t^{2}z_{2}-y_{3}\right) \\
\times\left( t^{2}y_{2}-y_{1}\right) \left( t^{3}y_{2}-y_{1}\right)
\left( t^{4}y_{2}-y_{1}\right) \left( t^{3}y_{3}-y_{1}\right) \left(
t^{4}y_{3}-y_{1}\right) \left( t^{2}y_{3}-y_{2}\right) \left( t^{3}%
y_{3}-y_{2}\right) .
\end{multline*}

\subsubsection{The general polynomials of isotype $\tau$}
In this section, the previous specialization and factorization results will be extended to arbitrary polynomials of isotype $\tau$. Specifically, we consider polynomials with property SMP. The factorization and vanishing properties will be combined with the known factorizations of highest weight symmetric polynomials by means of
the projection map to obtain expressions consisting of purely linear factors.

We start by applying our previous results. Let $\left\{g_{\S}(x):\ \S\in\tab_\tau\right\} $ be a basis for a space of polynomials of isotype $\tau$. By Proposition \ref{gdiv}, $g_{\S_{0}}(x)$ is divisible by $h_{\S_{0}}(x)$ and $\left\{ g_{\S}\right\} $ satisfies the same transformation properties as $\left\{h_{\S}\right\} $, in particular the one described in \eqref{hxs}. Finally, Proposition \ref{vstep} and Theorem \ref{hS0start} also apply to $\left\{ g_{\S}\right\}$, and so we have the following result.
\begin{thm}\label{gSSzero}
Let $\S\in\tab_\tau $ and $\Y\in\rst_\tau$ such that $\Y\neq \S$ and $\inv_{0}\left(\Y\right) \leq\inv_{0}\left(\S\right)$. Then $g_{\S}\left(\overline{x}(\Y)\right) =0$.
Furthermore, if $\S\neq \S_1$, then $\inv_0(\S)\neq 0$ and $g_{\S}\left(\overline{x}\left(\S_{1}\right)\right) =0$.
\end{thm}
The analog of Theorem \ref{heqfS} follows using the same argument.
\begin{thm}
For $\S\in\tab_\tau$, 
\begin{eqnarray*}
g_{\S}\left(\overline{x}(\S)\right) =\dfrac{g_{\S_{0}}\left( \overline{x}\left(\S_{0}\right)\right) }{f_{\S_{0}}\left( \overline{x}\left(\S_{0}\right) \right) } \cdot f_{\S}\left( \overline{x}(\S)\right).
\end{eqnarray*}
\end{thm}
Note the specific formula for $g_{\S_{1}}\left( \overline{x}\left(\S_{1}\right) \right) =\dfrac{g_{S_{0}}\left( \overline{x}\left(\S_{0}\right) \right)}{f_{\S_{0}}\left( \overline{x}\left(\S_{0}\right) \right) }\cdot f_{\S_{1}}\left( \overline{x}\left(\S_{1}\right)\right)$, which follows from Proposition \ref{fS1S1}. Thus, a factorization of $\dfrac{g_{\S_{0}}\left( \overline{x}\left(\S_{0}\right)\right) }{f_{\S_{0}}\left( \overline{x}\left(\S_{0}\right) \right) }$ provides one for $g_{\S_{1}}\left(\overline{x}\left(\S_{1}\right) \right) $, or conversely.

Here is an example in terms of some singular Macdonald polynomials. The polynomial $M_{\left(2,2,2,0,0,0\right)}$ is singular for $qt^{2}=-1$ and of isotype $\left(3,3\right)$. This polynomial corresponds to $\S_{1}$ while $M_{\left( 2,0,2,0,2,0\right) }$, with $qt^{2}=-1$, corresponds to $\S_{0}$, since $\alpha(\S_0)=(2,0,2,0,2,0)$ and $\alpha(\S_1)=(2,2,2,0,0,0)$. In this case, the special points are $\overline{x}\left(\S_{0}\right)=\left( z_{1},y_{1}t^{-2},z_{2},y_{1}t^{-1},z_{3},y_{1}\right) $ and $\overline{x}\left(\S_{1}\right) =\left(z_{1},z_{2},z_{3},y_{1}t^{-2},y_{1}t^{-1},y_{1}\right)$. Also,
\begin{multline*}
f_{\S_{0}}\left( x\right)= \left( tx_{1}-x_{2}\right) \left(tx_{3}-x_{4}\right) \left( tx_{5}-x_{6}\right) \rightsquigarrow f_{\S_{0}}\left( \overline{x}\left(\S_{0}\right) \right)=\frac{\left( z_{1}t^{3}-y_{1}\right) \left( z_{2}t^{2}-y_{1}\right)\left( z_{3}t-y_{1}\right)}{t^3}, \\
f_{\S_{1}}\left( x\right) = \left( tx_{1}-x_{4}\right) \left(t^{2}x_{2}-x_{5}\right) \left( t^{3}x_{3}-x_{6}\right)\rightsquigarrow
f_{\S_{1}}\left( \overline{x}\left(\S_{1}\right) \right)=\frac{\left( t^{3}z_{1}-y_{1}\right) \left( t^{3}z_{2}-y_{1}\right)\left( t^{3}z_{3}-y_{1}\right)}{t^3} .
\end{multline*}
By computing $M_{\alpha}$ (using the Yang-Baxter graph, \cite{DL2011}, for instance) and specializing it at $q=-t^{-2}$, we get that
\begin{eqnarray*}
\frac{M_{\left( 2,0,2,0,2,0\right) }\left( \overline{x}\left(\S_{0}\right) \right) }{f_{\S_{0}}\left( \overline{x}\left(\S_{0}\right)\right) }=\frac{M_{\left( 2,2,2,0,0,0\right) }\left( \overline{x}\left(\S_{1}\right) \right) }{f_{\S_{1}}\left( \overline{x}\left(\S_{1}\right)\right) }=-t^{3}\left( y_{1}+tz_{1}\right) \left( y_{1}+tz_{2}\right)\left( y_{1}+tz_{3}\right) .
\end{eqnarray*}
Although, in general, there is no linear factorization of $\dfrac{g_{\S_{0}}\left( \overline{x}\left(\S_{0}\right) \right) }{f_{\S_{0}}\left( \overline{x}\left(\ S_{0}\right) \right) }$ necessarily.

\subsubsection{Application to projections and factorizations}

Let us consider the polynomials $M_\mu$ with 
\begin{eqnarray*}
\mu=\left( \left( \left( d+K-1\right) m\right) ^{\nu_{K}},\left(\left( d+K-2\right) m\right) ^{n-1},\ldots\left( dm\right) ^{n-1},0^{dn-1}\right) ,
\end{eqnarray*}
with $n\geq2$, $1\leq\nu_{K}\leq n-1$ and $N=\ell(\mu) = dn-1+(n-1)\left( K-1\right) +\nu_{K}$. The associated isotype is $\tau=\left( dn-1,\left( n-1\right) ^{K-1},\nu_{K}\right)$. We also consider the specialization of the parameters $q$ and $t$, $\varpi=\left( q,t\right) =\left( \omega u^{-n/g},u^{m/g}\right) $ where $g=\gcd\left( m,n\right) $ and $\omega^{m/g}$ is a primitive $g^{\text{th}}$ root of unity. 

In Section \ref{SubHeckealgebra}, we show that for quasistaircase partitions and their associated $\tau$, there is a linear space of singular polynomials of isotype $\tau$ for $\mathcal{H}_{N}\left(t\right) $, $\left\{M_{\alpha\left(\S\right) }\ |\ \S\in\tab_\tau \right\} $.  

For $\S_{1}$, we use the following notation $\nu_{0}=N$ and $\nu_{i+1}=N-\left( d-1\right) n-i\left( n-1\right) $, for $1\leq i<K$, so then $\alpha\left(\S_{1}\right)_j=0$ for $\nu_{1}<j\leq N$, and $\alpha\left(\S_{1}\right) =\left( d+i-1\right) m$, for $\nu_{i+1}<j\leq\nu_{K}$.
\begin{defi}
For free variables $\left(z_{1},\ldots,z_{\nu_{K}},y_{1},\ldots,y_{K}\right) $, we define the \emph{special point associated to $\S_1$} by $\overline{\overline{x}}_j:= \overline{x}(\S_1)_j = t^{j-\nu_{i+1}-1}y_{K-i}$, for $\nu_{i+1}<j\leq\nu_{i}$, and $z_j$, otherwise. 
\end{defi}
For instance, take parameters $N=11$, $d=2$, $n=3$ and $m=2$. Then,  $\mu=\left( 8,8,6,6,4,4,0^{5}\right)$, $\tau=\left( 5,2,2,2\right)$, $[\nu_0,\nu_1,\nu_2,\nu_3]=\left[ 11,6,4,2\right]$, and $\overline{\overline{x}}=\left( z_{1},z_{2},y_{1},ty_{1},y_{2},ty_{2},y_{3},ty_{3},t^{2}y_{3},t^{3}y_{3},t^{4}y_{3}\right)$.

For the projection application, let $\left\{ \widehat{h}_{\S}:\S\in\tab_\tau \right\} $ denote the basis of minimal degree polynomials of isotype $\tau$ for $\mathcal{H}_{N}\left( 1/t\right) $. The minimal polynomials correspond to reverse lattice permutation of $\lambda =\left( \left( K-1\right) ^{\nu_{K}},\left( K-2\right) ^{n-1},\left(K-3\right) ^{n-1},\ldots,1^{n-1},0^{dn-1}\right)$. 

We show that $F_{\tau} \rho=\sum\limits_{\S\in\tab_\tau}\dfrac{1}{\gamma\left(\S;t\right) }\widehat{h}_{\S}M_{\alpha\left(\S\right)}$ is a highest weight symmetric Macdonald polynomial of index $\lambda+\mu$, for parameters $\varpi$. Now, we evaluate $F_\tau \rho$ at the special point $\overline{\overline{x}}$. By Theorem \ref{gSSzero}, $M_{\alpha\left(\S\right) }\left(\overline{\overline{x}}\right) =0$ for $\S\neq \S_{1}$, and  therefore
\begin{eqnarray}\label{FhM}
F_{\tau}\rho\left(\overline{\overline{x}}\right) =\dfrac{1}{\gamma\left(\S_{1};t\right) }\widehat{h}_{\S_{1}}\left(\overline{\overline{x}}\right) M_{\alpha(\S_1)}\left( \overline{\overline{x}}\right) .
\end{eqnarray}
Notice that we do not claim that $\widehat{h}_{\S}\left( \overline{\overline{x}}\right) =0$ for $\S\neq \S_{1}$. In fact, this does not hold because $\widehat{h}_{\S}$ is associated with $\mathcal{H}_{N}\left(1/t\right)$. However, since $\widehat{h}_{\S_{1}}$ is symmetric for the variables in each row of $\S_{1}$, the evaluation formula applies with $t$ replaced by $1/t$ in $\overline{\overline{x}}$ because the substring $\left[t^{-k}y,t^{1-k}y,\ldots,y\right]$ is transformed to $\left[t^{k}y,t^{k-1}y,\ldots,y\right]$ which is a permutation of $u=\left[ t^{k}y^{\prime},t^{k-1}y^{\prime},\ldots y^{\prime}\right] $ where $y^{\prime}=t^{k}y$, and $u$ equals the appropriate substring of $\overline{\overline{x}}$ with a changed variable. In fact, if $u$ comes from the $i^{\text{th}}$ row of $\S_{1}$ and $k=\tau_{i}-1$, then $y_{i}$ is replaced by $y_{i}^{\prime}=t^{\tau_{i}-1}y_{i}$. This procedure leads to the desired evaluation of $\widehat{h}_{\S_{1}}\left(\overline{\overline{x}} \right) $. We also need to use the partial factorization of $M_{\mu}\left( \overline{\overline{x}}\right)$, proved in the previous section, since $M_{\mu}$ is of the form $g_{\S_{1}}$ when $\left(q,t\right) =\varpi$.

In the paper \cite[Theorem 7.3]{CDL2017}, we provide a linear factorization of $F_{\tau}\rho\left(\overline{\overline{x}}\right) $. To state the result introduce the product
\begin{eqnarray*}
\left(a,b;q\right)_{k}:=\prod_{i=1}^{k}\left( a-bq^{i-1}\right),
\end{eqnarray*}
and for the parameters $\left( m,n,K,\nu_{K};q,t\right) $, define
\begin{multline*}
G\left( m,n,K,\nu_{K};q,t;z,y\right)  :=\prod_{u=1}^{\nu_{K}}\left( z_{u},t^{dn-1}y_{K};q\right) _{dm+1}\prod_{i=1}^{K-1}\left(z_{u},t^{n-1}y_{i};q\right) _{m+1}\\
\times\prod_{i=1}^{K-1}\prod_{s=0}^{n-2}\left( t^{s}y_{i},t^{dn-1}y_{K};q\right) _{dm+1}\times\prod_{1\leq i<j\leq K-1}\prod_{s=0}^{n-2}\left( t^{s}y_{i},t^{n-1}y_{j};q\right) _{m+1}.
\end{multline*}
The last product has the $y$-indices reversed from the original statement \cite[Theorem 7.3]{CDL2017}. That
formula is not symmetric in $\left( y_{1},\ldots,y_{K-2}\right)$, for general $\left( q,t\right) $, unless $q^{m}t^{n}=1$. Here the purpose is to demonstrate the role of $\widehat{h}_{\S_{1}}\left(\overline{\overline{x}}\right) h_{\S_{1}}\left(\overline{\overline{x}}\right) $ in the factorization. Neither of $h_{\S_{1}}$ and $\widehat{h}_{\S_{1}}$ are symmetric in $\left( y_{1},\ldots,y_{K-1}\right) $ but the product of them is. The result \cite[Theorem 7.3]{CDL2017} translates as $F_{\tau}\rho\left(\overline{\overline{x}}\right) \displaystyle{\mathop=^{(*)}}G\left( m,n,K,\nu_{K};q,t;z,y\right)$. 

Our next goal is deriving formulas for $h_{\S_{1}}\left(\overline{\overline{x}} \right)$ and $\widehat{h}_{\S_{1}}\left(\overline{\overline{x}}\right)$. 
\begin{prop}
With parameters $\left( N,m,n,d\right) $ and $\left( K,\nu_{K}\right) $ and $\left( q,t\right) =\varpi$ as above,
\begin{multline*}
h_{\S_{1}}\left(\overline{\overline{x}}\right) \mathop =^{(*)}\prod_{k=1}^{\nu_{K}}\left( z_{k}-t^{dn-1}q^{dm}y_{K}\right) \prod_{i=1}^{K-1}\left( z_{k}-t^{n-1}q^{m}y_{i}\right)\prod_{i=1}^{K-1}\prod_{s=0}^{n-2}\left( t^{s}y_{i}-t^{dn-1}q^{dm}y_{K}\right) \\
\times\prod_{1\leq i<j\leq K-1}\prod_{s=0}^{n-2}\left( t^{s}y_{i}-t^{n-1}q^{m}y_{j}\right) .
\end{multline*}
\end{prop}

\begin{proof}
Since $\overline{\overline{x}}_j=\overline{x}(\S_1)_j$, we apply the factorization in Proposition \ref{fS1S1} to $h_{\S_1}\left(\overline{\overline{x}}\right)$. Moreover, rename the variables $y_{i}$ in that formula by $y_{i}^{\prime}$, so that we can do the change of coordinates $y_{K}=t^{2-dn}y_{1}^{\prime}$ and $y_{i}=t^{2-n}y_{K+1-i}^{\prime}$, for $1\leq i\leq K-1$. Furthermore, in this case, the parameters are $l=K+1,\tau_{l}=\nu_{K}$, $\tau_{1}=dn-1$ and $\tau_{i}=n-1$, for $2\leq i\leq K$. Therefore, the formula becomes
\begin{multline*}
h_{\S_{1}}\left(\overline{\overline{x}}\right)  \mathop =^{(*)}\prod_{k=1}^{\nu_{K}}\left( t^{dn-1}z_{k}-t^{dn-2}y_{K}\right) \prod_{i=2}^{K}\left( t^{n-1}z_{k}-t^{n-2}y_{K+1-i}\right)\times \\
 \times \prod_{k=1}^{n-1}\prod_{2\leq i<j\leq K}\left(t^{n-1-k+1+n-2}y_{K+1-j}-t^{n-2}y_{K+1-i}\right)  \times\prod_{j=2}^{K}\prod_{k=1}^{n-1}\left(t^{dn-1-k+1+n-2}y_{K+1-j}-t^{dn-2}y_{K}\right) \\
\mathop =^{(*)}\prod_{k=1}^{\nu_{K}}\prod_{i=1}^{K}\left( tz_{k}-y_{i}\right) \times\prod_{k=1}^{n-1}\prod_{1\leq i<j\leq K-1}\left( t^{n-k}y_{i}-y_{j}\right)\times\prod_{j=1}^{K-1}\prod_{k=1}^{n-1}\left(t^{n-k}y_{j}-y_{K}\right) .
\end{multline*}
Substituting $q^{m}=t^{-n}$ in the claimed formula, this agrees up to a power of $t$ with the latter formula.
\end{proof}

\begin{prop}
With parameters $\left( N,m,n,d\right) $ and $\left( K,\nu_{K}\right) $ and $\left( q,t\right) =\varpi$ as above,
\begin{multline*}
\widehat{h}_{\S_{1}}\left(\overline{\overline{x}}\right)  \mathop =^{(*)}\prod_{k=1}^{\nu_{K}}\left( z_{k}-t^{dn-1}y_{K}\right) \prod_{i=1}^{K-1}\left( z_{k}-t^{n-1}y_{i}\right) \prod_{i=1}^{K-1}\prod_{s=0}^{n-2}\left( t^{s}y_{i}-t^{dn-1}y_{K}\right) \times\\ \times \prod_{1\leq i<j\leq K-1}\prod_{s=0}^{n-2}\left( t^{s}y_{i}-t^{n-1}y_{j}\right) .
\end{multline*}
\end{prop}

\begin{proof}
First, we consider the changing the variables $y_{i}\rightarrow t^{2-n}y_{i}$, for $1\leq i\leq K-1$, and $y_{K}\rightarrow t^{2-dn}y_{K}$. and so
\begin{multline*}
\widehat{h}_{\S_{1}}\left(\overline{\overline{x}}\right)  \mathop =^{(*)}\prod\limits_{k=1}^{\nu_{K}}\left( tz_{k}
-t^{2-dn}y_{K}\right) \prod\limits_{i=1}^{K-1}\left( tz_{k}-t^{2-n} y_{i}\right) \times\prod\limits_{k=1}^{n-1}\prod\limits_{1\leq i<j\leq K-1}\left( t^{n-k+2-n}y_{i}-t^{2-n}y_{j}\right) \\
\times\prod\limits_{j=1}^{K-1}\prod\limits_{k=1}^{n-1}\left(
t^{n-k+2-n}y_{j}-t^{2-dn}y_{K}\right) .
\end{multline*}
Now, we change $t\rightarrow t^{-1}$ to get
\begin{multline*}
\widehat{h}_{\S_{1}}\left(\overline{\overline{x}}\right)\mathop=^{(*)} \prod\limits_{k=1}^{\nu_{K}}\left( z_{k}-t^{dn-1}y_{K}\right) \prod\limits_{i=1}^{K-1}\left( z_{k}-t^{n-1}y_{i}\right) \prod\limits_{k=1}^{n-1}\prod\limits_{1\leq i<j\leq K-1}\left(y_{i}-t^{n-k}y_{j}\right) \times \\ \times \prod\limits_{j=1}^{K-1}\prod\limits_{k=1}^{n-1}\left( y_{j}-t^{dn-k}y_{K}\right).
\end{multline*}
This agrees up to a power of $t$ with the stated formula, where the index $s$, with $0\leq s\leq n-2$, is changed to $k=s+1$, so that $1\leq k\leq n-1$.
\end{proof}

Recall that $\left\{M_{\alpha\left(\S\right) }\ |\ \S\in\tab_\tau \right\} $ is a basis for isotype $\tau$ polynomials when $\left( q,t\right) =\varpi$. Then, as a consequence of the formula for $F_{\tau}\rho\left(\widehat{x}\left(\S_{1}\right) \right) $ from \cite{CDL2017} and \eqref{FhM}, we derive the following complete factorization for any $\S\in\tab_\tau$.
\begin{cor}
\begin{eqnarray*}
M_{\alpha\left(\S\right) }\left( \overline{x}\left(\S\right);\varpi\right) =\frac{M_{\mu}\left(\overline{\overline{x}};\varpi\right)}{h_{\S_{1}}\left(\overline{\overline{x}}\right) } \cdot f_{\S}\left( \overline{x}\left(\S\right) \right),
\end{eqnarray*}
where we also have that
\begin{multline*}
\frac{M_{\mu}\left(\overline{\overline{x}}\right)}{h_{\S_{1}}\left(\overline{\overline{x}} \right) }\mathop = ^{(*)}\prod_{u=1}^{\nu_{K}}\left( z_{u},t^{dn-1}qy_{K};q\right) _{dm-1}\prod_{i=1}^{K-1}\left( z_{u},t^{n-1}qy_{i};q\right) _{m-1}\\
 \times\prod_{i=1}^{K-1}\prod_{s=0}^{n-2}\left( t^{s}y_{i},t^{dn-1}qy_{K};q\right) _{dm-1}\times\prod_{1\leq i<j\leq K-1}\prod_{s=0}^{n-2}\left( t^{s}y_{i},t^{n-1}qy_{j};q\right)_{m-1}.
\end{multline*}
\end{cor}

\section{Conclusion and perspective\label{sec:conclusion}}

By appealing to the theory of vector-valued Jack and Macdonald polynomials, we have demonstrated the close relationship between highest weight symmetric polynomials and nonsymmetric singular polynomials. The relationship depends entirely on certain representations of the symmetric groups or of the Hecke algebra. Picturesquely, these representations come from tableaux formed by stacking the steps of the quasistaircase on top of each other.

In previous work, the authors have studied clustering properties of symmetric Macdonald polynomials, that is, the factorization into linear factors of such a polynomial specialized to certain sets of points. These points typically have a number of free variables. The present work should be useful in finding and proving clustering properties of nonsymmetric Macdonald polynomials by
exploiting the projection relationship.

Clustering properties of certain (symmetric homogeneous) Jack polynomials are of great interest for physicist in particular in the Quantum Hall Effect theory \cite{BH2008,BH2}. The tools developed in this paper as well as in the previous one \cite{CDL2017} contribute to a better knowledge of the properties highlighted by Bernevig and Haldane. It is worth noting that these properties are illuminated when studying generalized versions of Jack's polynomials: nonsymmetric, shifted, and $q$-deformed. In this paper, we have also used a more general variety of polynomials whose coefficients belong to irreducible representations of the Hecke algebra: the vector valued Macdonald polynomials. This shows that there is still a need to develop an arsenal of theoretical tools in order to fully understand the observations of physicists and in particular if we want to be able to properly prove the two remaining conjectures of Bernevig and Haldane. This is a broad research program  that we will continue to  explore in future works. It is very likely that the vector-valued Macdonald  polynomials will play a central role.

One of our goals is to promote, among physicists, the use of $q$-deformations and vector-valued polynomials. If we could convince them that these generalizations allow us to understand the fine properties of the wave functions they  investigate, then we would consider that we would have succeeded in our mission.

%%%%%%%%%%%%%%%%%%%%%%%%%%%%%%%%%%%%%%%%%%%%%%%%%%%%%%%%%%%%%%%%%%%%%%%%%%%%%%%%%
\bibliographystyle{plain}  
\bibliography{paper2-bib}

\begin{thebibliography}{10}

\bibitem{BF}
T.~H. Baker and P.~J. Forrester.
\newblock Nonsymmetric {J}ack polynomials and integral kernels.
\newblock {\em Duke Math. J.}, 95(1):1--50, 1998.

\bibitem{BBL}
H.~Belbachir, A.~Boussicault, and J.-G. Luque.
\newblock Hankel hyperdeterminants, rectangular {J}ack polynomials and even
  powers of the {V}andermonde.
\newblock {\em J. Algebra}, 320(11):3911--3925, 2008.

\bibitem{BH2008}
B.~A. Bernevig and F.D.M. Haldane.
\newblock Model fractional quantum hall states and {J}ack polynomials.
\newblock {\em Phys. Rev. Lett.}, 100:246802, 2008.

\bibitem{BH2}
B.~A. Bernevig and F.D.M. Haldane.
\newblock Clustering properties and model wave functions for non-abelian
  fractional quantum hall quasielectrons.
\newblock {\em Phys. Rev. Lett.}, 103, 2009.

\bibitem{BL}
A.~Boussicault and J.-G. Luque.
\newblock Staircase {M}acdonald polynomials and the {$q$}-discriminant.
\newblock In {\em 20th {A}nnual {I}nternational {C}onference on {F}ormal
  {P}ower {S}eries and {A}lgebraic {C}ombinatorics ({FPSAC} 2008)}, Discrete
  Math. Theor. Comput. Sci. Proc., AJ, pages 381--392. Assoc. Discrete Math.
  Theor. Comput. Sci., Nancy, 2008.

\bibitem{BLT}
A.~Boussicault, J.-G. Luque, and C.~Tollu.
\newblock Hyperdeterminantal computation for the {L}aughlin wavefunction.
\newblock {\em J. Phys. A}, 42(14):145301, 13, 2009.

\bibitem{Cherednik}
I.~Cherednik.
\newblock Nonsymmetric {M}acdonald polynomials.
\newblock {\em Internat. Math. Res. Notices}, (10):483--515, 1995.

\bibitem{CD2019}
L.~Colmenarejo and C.~F. Dunkl.
\newblock Singular nonsymmetric {M}acdonald polynomials and quasi-staircases.
\newblock In preparation.

\bibitem{CDL2017}
L.~Colmenarejo, C.~F. Dunkl, and J.-G. Luque.
\newblock Factorizations of symmetric {M}acdonald polynomials.
\newblock {\em Symmetry}, 10(Issue 11, 541), 2018.

\bibitem{dFGIL}
P.~Di~Francesco, M.~Gaudin, C.~Itzykson, and F.~Lesage.
\newblock Laughlin's wave functions, {C}oulomb gases and expansions of the
  discriminant.
\newblock {\em Internat. J. Modern Phys. A}, 9(24):4257--4351, 1994.

\bibitem{DJ1986}
R.~Dipper and G.~James.
\newblock Representations of {H}ecke algebras of general linear groups.
\newblock {\em Proc. London Math. Soc. (3)}, 52(1):20--52, 1986.

\bibitem{Dunkl1}
C.~F. Dunkl.
\newblock Differential-difference operators associated to reflection groups.
\newblock {\em Trans. Amer. Math. Soc.}, 311(1):167--183, 1989.

\bibitem{D2005}
C.~F. Dunkl.
\newblock Singular polynomials for the symmetric groups.
\newblock {\em Int. Math. Res. Not.}, (67):3607--3635, 2004.

\bibitem{DL2011}
C.~F. Dunkl and J.-G. Luque.
\newblock Vector-valued {J}ack polynomials from scratch.
\newblock {\em SIGMA Symmetry Integrability Geom. Methods Appl.}, 7:Paper 026,
  48, 2011.

\bibitem{DL2012}
C.~F. Dunkl and J.-G. Luque.
\newblock Vector valued {M}acdonald polynomials.
\newblock {\em S\'em. Lothar. Combin.}, 66:Art. B66b, 68, 2011/12.

\bibitem{DL2015}
C.~F. Dunkl and J.-G. Luque.
\newblock Clustering properties of rectangular {M}acdonald polynomials.
\newblock {\em Ann. Inst. Henri Poincar\'e D}, 2(3):263--307, 2015.

\bibitem{FJMM1}
B.~Feigin, M.~Jimbo, T.~Miwa, and E.~Mukhin.
\newblock A differential ideal of symmetric polynomials spanned by {J}ack
  polynomials at {$\beta=-(r-1)/(k+1)$}.
\newblock {\em Int. Math. Res. Not.}, (23):1223--1237, 2002.

\bibitem{FJMM2}
B.~Feigin, M.~Jimbo, T.~Miwa, and E.~Mukhin.
\newblock Symmetric polynomials vanishing on the shifted diagonals and
  {M}acdonald polynomials.
\newblock {\em Int. Math. Res. Not.}, (18):1015--1034, 2003.

\bibitem{GGJL17}
S.~Griffeth, A.~Gusenbauer, D.~Juteau, and M.~Lanini.
\newblock Parabolic degeneration of rational {C}herednik algebras.
\newblock {\em Selecta Math. (N.S.)}, 23(4):2705--2754, 2017.

\bibitem{JL}
Th. Jolicoeur and J.-G. Luque.
\newblock Highest weight {M}acdonald and {J}ack polynomials.
\newblock {\em J. Phys. A}, 44(5):055204, 21, 2011.

\bibitem{KTW}
R.~C. King, F.~Toumazet, and B.~G. Wybourne.
\newblock The square of the {V}andermonde determinant and its
  {$q$}-generalization.
\newblock {\em J. Phys. A}, 37(3):735--767, 2004.

\bibitem{Lascoux1}
A.~Lascoux.
\newblock Yang-{B}axter graphs, {J}ack and {M}acdonald polynomials.
\newblock {\em Ann. Comb.}, 5(3-4):397--424, 2001.
\newblock Dedicated to the memory of Gian-Carlo Rota (Tianjin, 1999).

\bibitem{Lascoux2}
A.~Lascoux.
\newblock {S}chubert and {M}acdonald polynomials, a parallel.
\newblock Notes available at
  http://igm.univ-mlv.fr/$\sim$al/ARTICLES/Dummies.pdf, 2008.

\bibitem{Lassalle}
M.~Lassalle.
\newblock Coefficients binomiaux g\'en\'eralis\'es et polyn\^omes de
  {M}acdonald.
\newblock {\em J. Funct. Anal.}, 158(2):289--324, 1998.

\bibitem{Laughlin}
B.~Laughlin.
\newblock Anomalous quantum hall effect: An incompressible quantum liquid with
  fractionally charged excitations.
\newblock {\em Phys. Rev. Lett.}, 50, 1983.

\bibitem{Luque}
J.-G. Luque.
\newblock {M}acdonald polynomials at {$t=q^k$}.
\newblock {\em J. Algebra}, 324(1):36--50, 2010.

\bibitem{Macd1995}
I.~G. Macdonald.
\newblock {\em Symmetric functions and {H}all polynomials}.
\newblock Oxford Mathematical Monographs. The Clarendon Press, Oxford
  University Press, New York, second edition, 1995.

\bibitem{Macdonald}
I.~G. Macdonald.
\newblock Affine {H}ecke algebras and orthogonal polynomials.
\newblock {\em Ast\'erisque}, (237):Exp.\ No.\ 797, 4, 189--207, 1996.
\newblock S\'eminaire Bourbaki, Vol. 1994/95.

\bibitem{Marshall}
D.~Marshall.
\newblock Symmetric and nonsymmetric {M}acdonald polynomials.
\newblock {\em Ann. Comb.}, 3(2-4):385--415, 1999.
\newblock On combinatorics and statistical mechanics.

\bibitem{MR}
G.~Moore and N.~Read.
\newblock Nonabelions in the fractional quantum {H}all effect.
\newblock {\em Nuclear Phys. B}, 360(2-3):362--396, 1991.

\bibitem{Opdam}
E.~M. Opdam.
\newblock Harmonic analysis for certain representations of graded {H}ecke
  algebras.
\newblock {\em Acta Math.}, 175(1):75--121, 1995.

\bibitem{RR}
N.~Read and E.~Rezayi.
\newblock Beyond paired quantum hall states: Parafermions and incompressible
  states in the first excited landau level.
\newblock {\em Phys. Rev. Lett.}, 59, 1999.

\bibitem{TSW}
T.~Scharf, J.-Y. Thibon, and B.~G. Wybourne.
\newblock Powers of the {V}andermonde determinant and the quantum {H}all
  effect.
\newblock {\em J. Phys. A}, 27(12):4211--4219, 1994.

\end{thebibliography}
%\bibliography{biblio_Fact}

\end{document}